\documentclass[11pt]{article}
\usepackage[a4paper, margin=2.5cm]{geometry}

\usepackage{graphicx}
\usepackage{amsmath,amsthm,amsfonts,amscd}
\usepackage{xcolor}
\usepackage{booktabs}
\usepackage{url}
\usepackage{hyperref}

\usepackage{thmtools, thm-restate}
\usepackage{enumerate}
\usepackage{enumitem}
\usepackage[ruled]{algorithm}
\usepackage[commentColor=black]{algpseudocodex}
\usepackage{tabularray}
\usepackage{float}

\usepackage{tikz}
\usetikzlibrary{arrows.meta}
\usetikzlibrary{calc, arrows}
\usetikzlibrary{shapes.multipart, positioning}

\usepackage{xspace}
\newcommand{\congest}{CONGEST\xspace}
\newcommand{\whp}{w.h.p.\xspace}
\newcommand{\CC}[3]{\ensuremath{CC_{#3}(#1,#2)}}
\newcommand{\nil}{\textsc{Nil}}

\newcommand{\hide}[1]{}

\newcommand{\lineComment}[1]{\LComment{ #1}}
\newcommand{\rightComment}[1]{\Comment{ #1}}

\begin{document}

\newtheorem{theorem}{Theorem}
\newtheorem{observation}[theorem]{Observation}
\newtheorem{corollary}[theorem]{Corollary}
\newtheorem{lemma}[theorem]{Lemma}
\newtheorem{definition}{Definition}

\theoremstyle{remark}
\newtheorem{note}{Note}
\newtheorem*{problem}{Problem}

\title{Distributed Distance Sensitivity Oracles\footnote{Authors' affiliation: Department of Computer Science, The University of Texas at Austin, Austin, TX, USA; email:  {\tt vigneshm@cs.utexas.edu, vlr@cs.utexas.edu}. This work was supported in part by NSF grant CCF-2008241.}
\author {Vignesh Manoharan  and Vijaya Ramachandran
}
}
\date{}

\maketitle

\begin{abstract}
 We present results for the distance sensitivity oracle (DSO) problem, where one needs to preprocess a given directed weighted graph $G=(V,E)$ in order to answer queries about the shortest path distance in $G$ from vertex $s$ to vertex  $t$ avoiding edge $e$, for any $s,t \in V, e \in E$. 
    DSO enables optimal re-routing under a link failure, and can serve as a key component for fault tolerance in a distributed setting. 
     However, no non-trivial results for DSO are known in the distributed CONGEST model.
     
    We present DSO algorithms with different tradeoffs between preprocessing and query cost: one that optimizes query response rounds, and another that prioritizes preprocessing rounds. We complement these algorithms with unconditional CONGEST lower bounds for DSO. Our DSO lower bounds build on a lower bound we present for the $k$-source shortest paths problem ($k$-SSP), which may be of independent interest. Additionally, we present almost-optimal upper and lower bounds for the related all pairs second simple shortest path (2-APSiSP) problem. 
\end{abstract}

\section{Introduction}

In distributed networks, maintaining communication in the event of a link failure is crucial for fault-tolerance. We model a distributed communication network as a graph $G=(V,E)$, and investigate the problem of computing shortest path distance when an arbitrary edge fails. Formally, the Distance Sensitivity Oracle (DSO) problem requires preprocessing $G$ to build an oracle that can quickly answer queries of the form $d(x,y,e)$ for $x,y\in V, e\in E$: here, $d(x,y,e)$ is the shortest path distance from $x$ to $y$ when edge $e$ is removed from $G$. In the special case when $x$ and $y$ are fixed vertices and the $d(x,y,e)$ values are computed explicitly, this is known as the Replacement Paths (RPaths) problem.

Let $|V|=n, |E|=m$. A naive DSO algorithm to answer a query $d(x,y,e)$ would be to simply compute SSSP (Single Source Shortest Paths) from $x$ in the $G - \{e\}$, i.e., $G$ with edge $e$ removed. However, by first preprocessing the graph and storing certain distance information, we can answer queries much faster. 
In an extreme case, 
one could precompute and store all possible replacement path distances. However, this would require high preprocessing cost and oracle size to compute all $n^2m$ distances $d(x,y,e)$ (this can be reduced to $O(n^3)$ by restricting to $e$ on $x$-$y$ shortest path). The primary goal in designing DSO algorithms is to find good tradeoffs between preprocessing cost, query cost, and the oracle size.

In the sequential setting, the first algorithms for constructing a DSO were given in~\cite{DemetrescuTCR08dso}. Subsequent improvements~\cite{BernsteinK08dso,BernsteinK09dso} led to an algorithm that preprocesses a graph in $\tilde{O}(mn)$\footnote{We use the notation $\tilde{O}, \tilde{\Omega}, \tilde{\Theta}$ to hide poly-logarithmic factors in our costs.} time to construct an oracle of size $\tilde{O}(n^2)$ that can answer any query in $O(1)$ time.

 The DSO problem is particularly relevant for fault tolerance in the distributed setting, as it models re-establishing efficient communication after link failure occurs.
Despite this relevance, no non-trivial results for DSO were known in the distributed CONGEST model prior to our work. In this paper, we present the first algorithms and lower bounds for DSO tailored to the distributed CONGEST model. 

We introduce a batched DSO model for the distributed setting in which we preprocess an input graph to answer $k \ge 1$ queries of the form $d(x_i, y_i, e_i)$ for $x_i, y_i \in V, e_i \in E$. This generalizes the sequential DSO setup which has $k=1$, i.e., where queries are answered one at a time. The query cost is defined as the total time to answer all queries, as a function of $k$. At the end of the query algorithm, the distance $d(x_i, y_i, e_i)$ must be stored at $y_i$. Note that despite multiple queries, only one edge fails at a time, i.e., the graph is reset to its original state for any subsequent query.

Depending on the application, we consider two natural models for DSOs with batched queries. In the first model, all edges $e_i$ are the same across the $k$ queries, which models applications where we want to compute distances between different vertex pairs when a failure occurs; we denote this the \textit{single-edge batched DSO} or \textit{seb-DSO} for short. In the second model, the edges $e_i$ may be different across queries, and this models applications where we want to compute distances under distinct, independent edge failures; we denote this more general model as the \textit{general batched DSO} or simply refer to it as \textit{DSO} in the rest of the paper.

As noted earlier, a naive algorithm that performs no preprocessing can answer $k$ arbitrary queries in $O(k\cdot SSSP)$ rounds ($SSSP$ denotes the CONGEST round complexity of computing single source shortest paths). We present two CONGEST DSO algorithms that utilize preprocessing to beat this query cost: Our first algorithm optimizes query cost, requiring only $O(k+D)$ rounds to answer $k$ queries, and uses $\tilde{O}(n^{3/2})$ rounds for preprocessing. Our second algorithm prioritizes preprocessing cost, achieving $\tilde{O}(n)$ preprocessing rounds and answering $k$ queries in $\tilde{O}(k\sqrt{n}+D)$ rounds. Both algorithms achieve $\tilde{O}(n^2)$ oracle size, matching the size of the best $O(1)$-query sequential DSO~\cite{BernsteinK09dso}. 
In fact, both of our algorithms store only $\tilde{O}(n)$ pre-computed distances at each node. Both algorithms work for the general-DSO queries.

For the more specialized seb-DSO problem, we beat the naive $O(k\cdot SSSP)$ query cost even without preprocessing. We achieve this using a CONGEST algorithm for $k$-Source Shortest Paths ($k$-SSP) problem~\cite{ManoharanR24mwc}. In the $k$-SSP problem we are given a graph $G=(V,E)$ and a set of $k$ vertices $U \subseteq V$, and need to compute shortest path distances $d(u,v)$ for each $u \in U, v \in V$. This is a fundamental graph problem, and generalizes the important SSSP and APSP problems. In~\cite{ManoharanR24mwc}, it is shown that $k$-SSP can be computed faster than $O(k\cdot SSSP)$, and that result smoothly interpolates between the current best SSSP and APSP algorithms. We show in this paper that the round complexity of answering $k$ seb-DSO queries without preprocessing is equivalent to that of $k$-SSP, with both upper and lower bounds (discussed below). This gives us a $\tilde{\Theta}(\sqrt{nk} + D)$ round bound for $k \ge n^{1/3}$ that is nearly optimal, and a more involved expression for smaller $k$ where the current best upper and lower bounds do not match.

To complement our DSO algorithms, we present unconditional CONGEST lower bounds for computing answers to $k$ DSO queries. These lower bounds apply even for seb-DSO, in directed graphs (both weighted and unweighted) and undirected weighted graphs.
Our approach is to first establish a new CONGEST lower bound for the $k$-SSP problem, which is of independent interest.

We prove that computing $k$-SSP in directed unweighted graphs or in undirected weighted graphs requires $\tilde{\Omega}(\sqrt{nk}+D)$ rounds in CONGEST. This lower bound is nearly tight for $k \ge n^{1/3}$ in directed graphs, given the directed $k$-SSP algorithms in~\cite{ManoharanR24mwc} (which compute exact $k$-SSP in unweighted graphs and $(1+\epsilon)$-approximate $k$-SSP in weighted graphs). For undirected weighted graphs, the bound is almost tight for approximate $k$-SSP for all $k$, given the $(1+\epsilon)$-approximation $k$-SSP algorithm of~\cite{ElkinN19} that takes ${O}\left((\sqrt{nk}+D) \cdot n^{o(1)}\right)$ rounds. Note that undirected unweighted $k$-SSP can be solved in $O(k+D)$ rounds~\cite{LenzenPP19,HoangPDGYPR19}, which is readily seen to be optimal. Our result appears to be the first CONGEST lower bound between $n$ and $\sqrt{n}$ rounds for a shortest path problem.

Finally, we address the related All-Pairs Second Simple Shortest Paths (2-APSiSP) problem. For a pair of vertices $x,y \in V$, the Second Simple Shortest Path (2-SiSP) distance $d_2(x,y)$ is the minimum weight of a simple $x$-$y$ path that differs from an $x$-$y$ shortest path~\cite{Yen1971}. In the 2-APSiSP problem, we need to compute $d_2(x,y)$ for all pairs $x,y \in V$. This problem was studied in the sequential setting in~\cite{AgarwalR16sisp}. In this paper, we present the first nontrivial CONGEST algorithm to compute directed weighted 2-APSiSP. Our algorithm takes $\tilde{O}(n)$ rounds, and we show this is nearly optimal by presenting a matching lower bound even in undirected unweighted graphs. This lower bound continues to hold when APSP distances in the original graph are known beforehand, and even when the graph has constant diameter. 

A core procedure used in sequential DSO and 2-APSiSP algorithms~\cite{DemetrescuTCR08dso,BernsteinK09dso,AgarwalR16sisp} is that of computing excluded shortest paths distances~\cite{DemetrescuTCR08dso}, i.e., distances when certain types of paths are avoided one at a time. We present a distributed CONGEST procedure for this computation which we use in our 2-APSiSP algorithm and our first DSO algorithm.

\section{Preliminaries}
\label{sec:prelim}

\subsection{The CONGEST Model}\label{congest-model}
In the CONGEST model~\cite{Peleg2000book}, a communication network is represented by a graph $G=(V,E)$ where nodes model processors and edges model bounded-bandwidth communication links between processors. Each node has a unique identifier in $\{0, 1, \dots n-1\}$ where $n = |V|$. Initially, each node only knows the identifiers of itself and its neighbors in the network. Each node has infinite computational power. The nodes perform computation in synchronous rounds, where each node can send a message of up to $\Theta(\log n)$ bits to each neighbor and can receive the messages sent to it by its neighbors. The complexity of an algorithm is measured by the number of rounds until the algorithm terminates. 

We mainly consider directed weighted graphs $G$ in this paper, where each edge has an integer weight known to the vertices incident to the edge. Following the convention for CONGEST algorithms~\cite{Nanongkai14,CaoF23parallel}, the communication links are always bi-directional and unweighted. The undirected diameter of the network, which we denote by $D$, is an important parameter in the CONGEST model.

\subsection{Notation and Terminology}
\label{sec:notation}

We will deal primarily with a directed weighted graph $G=(V,E)$. Let $|V|=n$. Each edge $(x,y) \in E$ (for $x,y \in V$) has a non-negative integer weight $w(x,y) \le W$, with the maximum weight $W$ bounded by a polynomial in $n$. We denote the shortest path distance from $x$ to $y$ by $d(x,y)$, and a shortest path from $x$ to $y$ by $P_{xy}$. We use $G^r=(V,E^r)$ to denote the reversed graph of $G$, where each edge is oriented in the opposite direction. The undirected diameter $D$ of $G$ (and $G^r$), is the maximum shortest path distance between two vertices in the underlying undirected unweighted graph of $G$. 

Let $G-\{e\}$ be the graph $G$ with edge $e$ removed.
We use $d(x,y,e)$, for $x,y \in V, e \in E$, to denote the shortest path distance from $x$ to $y$ in the graph $G-\{e\}$, and we refer to this as the replacement path distance 
from $x$ to $y$ avoiding failed edge $e$.

In this context, we refer to $x$ as source, $y$ as target or sink, and $e$ as failed edge.
We generalize this definition for a path $P$, and $d(x,y,P)$ denotes the shortest path distance from $x$ to $y$ in $G-E(P)$.
We use $d_2(x,y)$ to denote the second simple shortest path distance from $x$ to $y$ (2-SiSP distance), which is the minimum distance of a simple $x$-$y$ path that differs from an $x$-$y$ shortest path $P_{xy}$ by at least one edge. Note that $d_2(x,y) = \min_{e \in P_{xy}} d(x,y,e)$. The distance $d_2(x,y)$ is independent of the chosen shortest path $P_{xy}$, and $d_2(x,y)=d(x,y)$ if there are multiple $x$-$y$ shortest paths.

\begin{definition}[Shortest path trees and independent paths]
We denote the out-shortest path tree rooted at a vertex $x \in V$ by $T_x$, and we refer to $x$ as its source. For a path $P$ in $T_x$, let $(u,v)$ be the first edge of $P$, so $u$ is the closest vertex to source $x$. We define the second vertex $v$ to be the \textit{key} vertex of $P$. We define $T_x(P)$ to be the subtree of $T_x$ rooted at the key vertex of $P$.
    
A set of paths $\mathcal{R}$ is called \textit{independent} with respect to source $x$
if each path $P \in \mathcal{R}$ is a subpath of $T_x$, and for any pair of distinct paths $P,P' \in \mathcal{R}$, the subtrees $T_x(P)$ and $T_x(P')$ are disjoint.
\label{def:indpaths}
\end{definition}

\subsection{Problem Definitions}
\label{sec:problemdefs} In the following definitions, $G=(V,E)$ is a directed weighted graph. 

\begin{definition}[Excluded Shortest Paths Problem]
    Given a set of sources $X \subseteq V$ and sets of independent paths $\mathcal{R}_x$ for each $x \in X$ (as defined in Section~\ref{sec:notation}), we need to compute $d(x,y,P)$ for each $x \in X, y \in V$ and $P \in \mathcal{R}_x$, i.e., the $x$-$y$ shortest path distance when path $P$ is removed. In algorithms for the \congest model, the input set of paths $\mathcal{R}_x$ is given as follows: Each edge $e\in E $ that belongs to some path $P \in \mathcal{R}_x$ knows the key vertex of $P$ (second vertex of $P$ from $x$).
    For the output, each distance $d(x,y,P)$ must be known at node $y$.
    \label{def:exclude}
\end{definition}

\begin{restatable}[Distance Sensitivity Oracles Problem (DSO)]{definition}{defdso}
In the DSO problem we need to compute $d(x,y,e)$, the $x$-$y$ shortest path distance when $e$ is removed, for any pair of vertices $x,y \in V$ and any $e \in E$. A DSO algorithm is allowed to perform some preprocessing on the graph $G$ to construct on oracle, which then answers queries of the form $d(x,y,e)$. In algorithms for the \congest model, we assume that at query time the query $d(x,y,e)$ is announced to a vertex incident to $e$. After the query is answered, distance $d(x,y,e)$ must be known at node $y$. \label{def:dso}
\end{restatable}

In sequential DSO algorithms, after preprocessing the input graph, queries are typically answered one at a time as in Definition~\ref{def:dso}. However, in a distributed DSO algorithm the oracle may be stored in a distributed manner across nodes and answering a query is no longer a single local lookup: up to $D$ rounds of communication may be needed to answer a single query, where $D$ is the undirected diameter of the CONGEST network. 

To alleviate the cost of $D$ per query, we introduce a batched DSO model where we answer $k > 1$ queries simultaneously. This allows us to take advantage of pipelining to improve total communication cost, thereby avoiding $\Omega(kD)$ rounds for $k$ queries. 
Batched DSO strictly generalizes Definition~\ref{def:dso} where $k=1$, and we address two models of batched DSO in this paper, as described in the following two definitions.

\begin{definition}[Single-Edge Batched DSO (seb-DSO)]
    The DSO algorithm must respond to $k$ queries $d(x_i,y_i,e)$, $1 \le  i \le k$ for $x_i,y_i \in V$, $e \in E$, where the failed edge $e$ is the same across all queries.
    \label{def:sebdso}
\end{definition}

\begin{definition}[General Batched DSO]
    The DSO algorithm must answer $k$ queries $d(x_i,y_i,e_i)$, $1 \le  i \le k$ for $x_i,y_i \in V$, $e_i \in E$, where each failed edge $e_i$ may be different. We simply use DSO to refer to this model in this paper.
    \label{def:generaldso}
\end{definition}

The following $k$-SSP problem  features in some of our algorithms as well as in our DSO lower bound. It was studied earlier in~\cite{ManoharanR24mwc}. 

\begin{definition}[$k$-Source Shortest Paths Problem ($k$-SSP)]
    Given a set of $k$ vertices $U \subseteq V$, compute the shortest path distances $d(u,v)$ for each $u \in U, v \in V$. In algorithms for the \congest model, each distance $d(u,v)$ must be known at node $v$.
    \label{def:kssp}
\end{definition}

The 2-APSiSP problem, defined below,  generalizes 2-SiSP~\cite{Yen1971}, and was studied earlier in~\cite{AgarwalR16sisp}. The 2-APSiSP problem is related to DSO and in this paper we present nearly matching upper and lower bounds for 2-APSiSP.

\begin{definition}[All Pairs Second Simple Shortest Path Problem (2-APSiSP)]
    For a pair of vertices $x,y \in V$, the Second Simple Shortest Path (2-SiSP) distance $d_2(x,y)$ is the minimum weight of a simple $x$-$y$ path that differs from an $x$-$y$ shortest path.\\
    In the 2-APSiSP problem we need to compute $d_2(x,y)$ for all pairs of vertices $x,y \in V$. In algorithms for the \congest model, each distance $d_2(x,y)$ ($\forall x \in V$) must be known at node $y$.
\end{definition}

\subsection{CONGEST Primitives}
\label{sec:primitives}
We state some basic CONGEST primitives used in our algorithms.

In an unweighted graph, a breadth-first search (BFS) from $k$ source vertices up to $h$ hops takes $O(k+h)$ rounds~\cite{LenzenPP19,HoangPDGYPR19}. We can broadcast $m$ messages ($O(\log n)$ bits each) to all other vertices in $O(m+D)$ rounds~\cite{Peleg2000book}. Given a rooted tree $T$ of height $t$ in the CONGEST network (where each node in the tree knows its incident edges), we can send $m$ messages ($O(\log n)$ bits each) down from root to all nodes in $T$ in $O(m+t)$ rounds -- this operation is called a downcast~\cite{Peleg2000book}, and we can even perform an associative operation on the received values at each node before propagating the information. A similar upcast can be done from the leaves to the root. A convergecast is an upcast along a broadcast tree that contains all vertices, which has height $D$, and thus takes $O(m+D)$ rounds~\cite{Peleg2000book}. 

\paragraph{Scheduling with random delays} Following~\cite{Ghaffari15scheduling,LeightonMR99}, we use scheduling with random delays in order to efficiently perform multiple computations on the network. Let the \textit{congestion} of a CONGEST algorithm be the maximum number of messages it sends through any single edge, over all rounds of the algorithm.
Consider an algorithm that runs in $R$ rounds and has congestion $C$. Then, we can run $k$ instances of the algorithm in $O(kC+R\log n)$ rounds, with total congestion $O(kC)$ over all instances.

\paragraph{$SSSP$ and $APSP$} We use $SSSP$ and $APSP$ to denote the round complexity in the CONGEST model for weighted single source shortest paths {\rm (SSSP)} and weighted all pairs shortest paths {\rm (APSP)} respectively. 
The current best algorithm for weighted {\rm APSP} runs in $\tilde{O}(n)$ rounds, randomized~\cite{BernsteinN19apsp}.
For unweighted graphs, $O(n)$ round deterministic APSP algorithms are known~\cite{HolzerW12,HoangPDGYPR19}. For weighted graphs, the current best deterministic APSP algorithm takes $\tilde{O}(n^{4/3})$~rounds~\cite{AgarwalR20apsp}.
Weighted {\rm SSSP} can be computed in  $\tilde{O}(\sqrt{n} + n^{2/5+o(1)}D^{2/5}+D)$ rounds by a  randomized algorithm~\cite{CaoF23parallel}. The current best lower bounds are $\Omega\left(\frac{\sqrt{n}}{\log n} + D\right)$ for weighted {\rm SSSP}~\cite{SarmaHKKNPPW12} and $\Omega\left(\frac{n}{\log n}\right)$ for (weighted and unweighted) {\rm APSP}~\cite{Nanongkai14}. In our algorithms, we also use a deterministic low-congestion SSSP procedure presented in~\cite{GhaffariT24sssp}, which takes $\tilde{O}(n)$ rounds but incurs only $\tilde{O}(1)$ congestion per edge. 
When computing shortest path distances from a source $x$, we assume that each vertex $v$ knows its parent on a shortest path from $x$ to $v$. This parent information can be computed in $O(1)$ additional rounds per source by having each vertex share its distance from $x$ with its neighbors. A vertex $v$ then identifies an in-neighbor $u$ as its parent if the condition $d(x,v) = d(x,u)+w(u,v)$ holds.

\section{Summary of Results}
\label{sec:dso:results}

In Table~\ref{tab:dso:results} we summarize our results for DSO (i.e., general-batch DSO) in the CONGEST model for a directed weighted graph $G=(V,E)$
with integer weights. Recall that $n=|V|$, and $D$ is the undirected diameter of $G$. Integer weights are needed in our algorithms only for computing SSSP using the algorithm in~\cite{GhaffariT24sssp}, otherwise our algorithms work for real weights, provided that $O(1)$ weights or $O(\log n)$ bits can be communicated in a single round along an edge.
All of  our algorithms are randomized and are correct \whp in $n$ (with probability $1-(1/poly(n))$).

\begin{table}[ht]
    \centering
        \begin{tblr}{hlines, vlines, columns={c}, column{1}={3cm}, column{2}={2.5cm}, column{3}={2.3cm}, column{4} = {2cm}, cells={m}}
            \textbf{CONGEST Result} & \textbf{Rounds to answer $k\ge 1$ queries} & \textbf{Preprocessing rounds} & \textbf{Space per node} \\ 
            {\textit{Naive algorithm}} & {$O(k \cdot SSSP)$} & {0} & {0} \\
            {\textit{Algorithm with fast \underline{query responses}} \\ (Theorem~\ref{thm:dso:algquery})} & {$O(k+D)$}  & {$\tilde{O}(n^{3/2})$} & {$\tilde{O}(n)$}  \\
            {\textit{Algorithm with fast \underline{preprocessing}} \\ (Theorem~\ref{thm:dso:algpre})} & {$\tilde{O}(k\sqrt{n}+D)$}  & {$\tilde{O}(n)$} & {$\tilde{O}(n)$}  \\
            {\textit{{Lower bound}, for any} $Q=o\left(\frac{\sqrt{n}}{\log n}\right)$ \\ (Theorem~\ref{thm:dso:prelb})} & {$O(k \cdot Q + D)$}  & {$\tilde{\Omega}\left(\frac{n}{Q}\right)$} & {-} \\
        \end{tblr}
        \caption{Summary of \congest DSO results (i.e. general-batch DSO) for a directed weighted graph $G=(V,E)$. $n=|V|$, and $D$ is the undirected diameter of $G$.}
        \label{tab:dso:results}
\end{table}

\subsection{Excluded Shortest Paths}
Our first DSO algorithm relies on a subroutine to efficiently compute shortest path distances from multiple sources when certain parts of shortest path trees are excluded, one at a time. We present an algorithm for the excluded shortest paths problem (see Definition~\ref{def:exclude}). 
Our CONGEST implementation is inspired by a sequential algorithm for a single source in~\cite{DemetrescuTCR08dso}, and we extend the result to multiple sources by using low congestion SSSP combined with efficient random scheduling. For brevity, we refer to this procedure as an exclude computation, which we present in detail in Section~\ref{sec:exclude}.

\begin{restatable}{theorem}{thmdsoexclude}
Given a set of sources $X \subseteq V$ and an independent set of paths $\mathcal{R}_x$ for each source $x \in X$, we can compute $d(x,y,P)$ for each $x \in X, y\in V, P\in R_x$ in $\tilde{O}(n)$ rounds. Additionally, the maximum congestion is $\tilde{O}(|X|)$. \label{thm:dso:exclude}
\end{restatable}

\subsection{Distance Sensitivity Oracles (DSOs)}

We summarize our DSO results here: Section~\ref{sec:dso:gendsointro} states our results for our general batched DSO algorithms, Section~\ref{sec:dso:sebintro} for our seb-DSO algorithms and Section~\ref{sec:results:lb} states our DSO lower bounds, including our $k$-SSP lower bound.

Our DSO algorithms focus on computing the distance value of a given query $d(x,y,e)$, but we can readily augment our algorithms so that each vertex on the replacement path from $x$ to $y$ avoiding $e$ knows the next vertex on that path.

\subsubsection{General Batched DSO Algorithms}
\label{sec:dso:gendsointro}

\paragraph{DSO with fast query responses}
Our first DSO optimizes query time, answering a batch of $k$ arbitrary queries $d(x,y,e)$ in $O(k + D)$ rounds. The algorithm performs exclude computations during preprocessing to store certain excluded distances at different vertices, such that we can combine appropriate distances to compute any $d(x,y,e)$. Once the query is known, we broadcast the relevant precomputed distances and compute $d(x,y,e)$ at $y$. To answer a batch of $k$ different queries, we pipeline the broadcasts for each query, answering them all in $O(k+D)$ rounds. 

A direct distributed implementation of a sequential algorithm from~\cite{DemetrescuTCR08dso} uses $\tilde{O}(n^{3/2})$ rounds for preprocessing and answers $k$ queries in $O(k+D)$ rounds, but requires an oracle size of $\tilde{O}(n^{3/2})$ per node. This result appears in our conference paper~\cite{ManoharanR25dso}, 

In this paper, we achieve the same $\tilde{O}(n^{3/2})$-round preprocessing and $O(k+D)$-round query bounds with an improved $\tilde{O}(n)$ oracle size per node. The total oracle size of $\tilde{O}(n^2)$ across all nodes matches the oracle size of the best sequential algorithm~\cite{BernsteinK09dso}. This improved construction uses ideas from a sequential DSO in~\cite{BernsteinK08dso} and is more involved due to the need for communicating distances and other information across different nodes during preprocessing while maintaining $\tilde{O}(n^{3/2})$ preprocessing rounds. This algorithm is described in Section~\ref{sec:dso:algquery}.

\begin{restatable}{theorem}{thmdsoalgquery}
We can construct a DSO for directed weighted graphs that takes $\tilde{O}(n^{3/2})$ preprocessing rounds, has $\tilde{O}(n)$ oracle size per node, and answers a batch of any $k$ queries of the form $d(x,y,e)$ in $O(k + D)$ rounds. \label{thm:dso:algquery}
\end{restatable}

\paragraph{DSO with fast preprocessing}
We present a DSO construction that prioritizes preprocessing by taking only $\tilde{O}(n)$ rounds after which it can answer $k$ queries in $\tilde{O}(k\sqrt{n}+D)$ rounds. Note that we can answer $k$ queries $d(x,y,e)$ without any preprocessing with $k$ SSSP computations in $\tilde{O}(k\cdot(\sqrt{n}+n^{2/5}D^{2/5}+D))$ rounds~\cite{CaoF23parallel}, and we improve on this bound for all $k\ge 1$. The algorithm stores $\tilde{O}(n)$ words per node, for a total of $\tilde{O}(n^2)$ space, which is competitive with the sequential algorithm of~\cite{BernsteinK09dso}.
We describe this algorithm in Section~\ref{sec:dso:algpre}.

In order to reduce preprocessing rounds to $\tilde{O}(n)$, we cannot perform too many exclude computations. So, we handle short replacement paths of hop-length $\le \sqrt{n}$ at query time using a $\sqrt{n}$-hop Bellman-Ford procedure. For longer replacement paths ($> \sqrt{n}$ hops), we construct randomly sampled graphs where each edge is included with certain probability (inspired by a sequential DSO in~\cite{WeimannY13}). In these graphs, we perform SSSP computations only from certain sampled vertices during preprocessing. This allows us to quickly compute long replacement path distances at query time. 

\begin{restatable}{theorem}{thmdsoalgpre}
We can construct a DSO for directed weighted graphs that takes $\tilde{O}(n)$ preprocessing rounds, has $\tilde{O}(n)$ oracle size per node, and then answers a batch of any $k$ queries of the form $d(x,y,e)$ in $\tilde{O}(k\sqrt{n} + D)$ rounds. \label{thm:dso:algpre}
\end{restatable}

\subsubsection{Single Edge Batched DSO}
\label{sec:dso:sebintro}

The algorithms discussed above work for general batched DSO. We now consider the seb-DSO model. We observe that with no preprocessing, the query cost of answering $k$ seb-DSO queries is equivalent to the complexity of computing $k$-SSP, with matching upper and lower bounds. The upper bound in Observation~\ref{thm:dso:seb} follows by observing that seb-DSO for a batch of  $k\leq n$ queries reduces to $k$-SSP; 
we also show a corresponding directed lower bound (Corollary~\ref{lem:seblb} in Section~\ref{sec:dso:lb}) that follows from our main DSO lower bound in Theorem~\ref{thm:dso:lb}. A similar fairly tight relation between $k$-SSP and seb-DSO holds for undirected graphs, given $k$-SSP algorithms in~\cite{ElkinN19,LenzenPP19,HoangPDGYPR19} and our undirected DSO lower bound (Section~\ref{sec:dso:lbunw}).

\begin{restatable}{observation}{thmdsoseb}
    Given a graph $G=(V,E)$ and $k$ seb-DSO queries $d(x_i, y_i, e)$ for $1 \le i \le k$ (where the failed edge $e$ is the same), we can answer the queries with no preprocessing in the same round complexity as computing $k$-SSP. 
        
    Using a $k$-SSP algorithm in~\cite{ManoharanR24mwc} yields the following bounds:
        
        \begin{enumerate}[topsep=0cm,label=\arabic*.,ref=\arabic*]
            \item In a directed \textit{unweighted} graph the query bound is:
                $$
                    \begin{cases}
                        \tilde{O}(\sqrt{nk} + D) & , \text{ if } k \ge n^{1/3} \\
                        \min\left( \tilde{O}\left(\frac{n}{k}+D\right), k \cdot SSSP \right) & , \text{ if } k<n^{1/3} 
                    \end{cases}
                $$
            \item In a directed \textit{weighted} graph, we can obtain $(1+\epsilon)$-approximate answers (for any constant $\epsilon>0$) for $k$ seb-DSO queries in rounds :
                $$
                    \begin{cases}
                        \tilde{O}(\sqrt{nk} + D) & , \text{ if } k \ge n^{1/3} \\
                        \tilde{O}(\sqrt{nk} + k^{2/5}n^{2/5+o(1)}D^{2/5} + D) & , \text{ if } k<n^{1/3}
                    \end{cases} 
                $$
        \end{enumerate}
\label{thm:dso:seb}
\end{restatable}

\subsubsection{Lower Bounds}
\label{sec:results:lb}

We start with stating a $k$-SSP lower bound that is a starting point for our DSO lower bound; this $k$-SSP lower bound may be of independent interest. We then state our DSO lower bounds.

\paragraph{$k$-SSP Lower Bound}
We present a lower bound of $\Omega\left(\frac{\sqrt{nk}}{\log n} + D\right)$ for computing $k$-SSP in directed unweighted graphs (it trivially applies for directed weighted graphs as well) and in undirected weighted graphs. Our bound also holds against approximation algorithms. In fact, we prove a stronger result that computing distances between $k$ given pairs of vertices requires $\Omega\left(\frac{\sqrt{nk}}{\log n} + D\right)$ rounds.

\begin{restatable}{theorem}{thmkssplb}
    \begin{enumerate}[label=\Alph*.,ref=\Alph*]
        \item Given a directed unweighted graph $G=(V,E)$, computing $k$-SSP requires $\Omega\left(\frac{\sqrt{nk}}{\log n} + D\right)$ rounds, even if $G$ has undirected diameter $\Theta(\log n)$. This lower bound holds for any constant approximation algorithm. \label{thm:kssp:lb:dir}
        \item Given an undirected weighted graph $G=(V,E)$, computing $k$-SSP requires $\Omega\left(\frac{\sqrt{nk}}{\log n} + D\right)$ rounds, even if $G$ has undirected diameter $\Theta(\log n)$. This lower bound holds for any $(2-\epsilon)$-approximation algorithm. \label{thm:kssp:lb:undir}
    \end{enumerate}
\label{thm:kssp:lb}
\end{restatable}

This lower bound is nearly tight for directed graphs in many cases: When $k \ge n^{1/3}$, it matches the $k$-SSP algorithm bounds in~\cite{ManoharanR24mwc} for exact directed unweighted $k$-SSP and $(1+\epsilon)$-approximate directed weighted $k$-SSP. For $k < n^{1/3}$, the bound is not tight even for approximate distances, as our algorithm takes $\tilde{O}(\sqrt{nk}+n^{2/5+o(1)}k^{2/5}D^{2/5}+D)$ rounds. This is not unexpected as we have a gap between the current best upper and lower bounds even for SSSP ($k=1$) whether exact or approximate, i.e., $\tilde{O}(\sqrt{n} + n^{2/5 + o(1)}D^{2/5}+D)$~\cite{CaoF23parallel} and $\tilde{\Omega}(\sqrt{n}+D)$~\cite{SarmaHKKNPPW12}. 

Our undirected weighted lower bound is almost tight for $(1+\epsilon)$-approximate $k$-SSP all $k$, given the algorithm of~\cite{ElkinN19} that takes ${O}\left((\sqrt{nk}+D) \cdot n^{o(1)}\right)$ rounds. For undirected unweighted $k$-SSP, the previously known algorithms taking $O(k+D)$ rounds~\cite{LenzenPP19,HoangPDGYPR19} are readily seen to be optimal.

\paragraph{DSO Lower Bound}
We build on the $k$-SSP lower bound in Theorem~\ref{thm:kssp:lb}.\ref{thm:kssp:lb:dir} to obtain a lower bound of $\Omega\left(\frac{\sqrt{nk}}{\log n} + D\right)$ to answer $k$ DSO queries in a directed graph, even if unweighted.  This lower bound applies even for seb-DSO, i.e., when all $k$ queries have the same failed edge (and hence trivially applies to general-DSO as well).

As our $k$-SSP lower bound in Theorem~\ref{thm:kssp:lb} also holds for undirected weighted graphs and holds against approximation algorithms, our new DSO lower bounds inherit these same properties: they apply to directed graphs for any constant approximation, and to undirected weighted graphs for $(2-\epsilon)$-approximation. For simplicity, we state the following lower bounds just for exact directed DSO. Section~\ref{sec:dso:lbunw} discusses our DSO lower bounds for the undirected weighted case.

\begin{restatable}{theorem}{thmdsolb}
Consider a DSO algorithm for directed unweighted graph $G$ that takes $P(n)$ rounds for preprocessing and then answers a batch of $k \le n/2$ queries $d(a_i,b_i,e_i)$ for $1\le i \le k$ in $Q(n,k)$ (for $k \le n/2$) rounds. Then $P(n) + Q(n,k)$ is  $\Omega\left(\frac{\sqrt{nk}}{\log n}\right)$. This lower bound applies even for seb-DSO, i.e., when all $e_i$ are the same, and even if $G$ has undirected diameter $\Theta(\log n)$.
\label{thm:dso:lb}
\end{restatable}

We apply Theorem~\ref{thm:dso:lb} in two ways: First,  by setting $P(n)=0$, we obtain 
an almost optimal lower bound for seb-DSO without preprocessing, complementing the upper bound result stated earlier in Observation~\ref{thm:dso:seb} (see Corollary~\ref{lem:seblb} in Section~\ref{sec:dso:lb}). 
Second, we assume query costs similar to Theorem~\ref{thm:dso:algquery} and Theorem~\ref{thm:dso:algpre} to derive the preprocessing lower bounds that complement our upper bounds as stated below in Theorem~\ref{thm:dso:prelb2}.
We also obtain a more general lower bound on preprocessing rounds that covers the salient range of query response complexities; this is discussed in Section~\ref{sec:dso:lb} (Theorem~\ref{thm:dso:prelb}).

\begin{restatable}{theorem}{thmdsoprelbb}
    Consider a DSO algorithm $\mathcal{A}$ for directed unweighted graph $G=(V,E)$ on $n$ nodes. \label{thm:dso:prelb2}
    \begin{enumerate}[label=\Alph*.]
        \item If a DSO algorithm can answer $k=\frac{n}{\log ^3 n}$ queries in $O(k+D)$ rounds, it must use $\Omega(\frac{n}{\log^3 n})$ preprocessing rounds.
        \item If a DSO algorithm can answer $k=\log n$ queries in $O(k\frac{\sqrt{n}}{\log^2 n}+D)$ rounds, it must use $\Omega(\frac{\sqrt{n}}{\log^2 n})$ preprocessing rounds.
    \end{enumerate}
\end{restatable}

There still remains a $\sqrt{n}$-gap in the preprocessing cost between our lower bounds and our DSO algorithms, for both the case of $O(1)$ query and $O(\sqrt{n})$ query. 

\subsection{All Pairs Second Simple Shortest Paths (2-APSiSP)}
We present near-optimal upper and lower bounds for the 2-APSiSP problem.

\subsubsection{2-APSiSP Algorithm}
We present a 
\congest algorithm for 2-APSiSP that runs in $\tilde{O}(n)$ rounds.
Our approach adapts a characterization from a sequential algorithm~\cite{AgarwalR16sisp} to the distributed setting, leveraging our excluded shortest paths primitive from Theorem~\ref{thm:dso:exclude}.
The algorithm performs one exclude computation from each vertex, followed by non-trivial local computation in order to compute all 2-APSiSP distances. We present this result in Section~\ref{sec:apsisp:ub}

\begin{restatable}{theorem}{thmapsispub}
We can compute 2-APSiSP for directed weighted graphs in $\tilde{O}(n)$ rounds. \label{thm:apsisp:ub}
\end{restatable}

\subsubsection{2-APSiSP Lower Bound}
We present a lower bound of $\tilde{\Omega}(n)$ for computing 2-APSiSP in undirected unweighted graphs, which proves that our $\tilde{O}(n)$ round algorithm is almost optimal for any class of graphs: directed or undirected, weighted or unweighted. 

An $\tilde{\Omega}(n)$ CONGEST lower bound for computing directed weighted 2-SiSP for one pair is shown in~\cite{ManoharanR24rp}, and trivially applies to 2-APSiSP as well. However, this lower bound applies only to directed weighted graphs and requires the input pair of vertices to have a $\Theta(n)$-hop shortest path. The lower bound in this paper for 2-APSiSP is stronger, as it applies to undirected unweighted graphs, even if the graph has constant diameter. We prove this result in Section~\ref{sec:apsisp:lb}.

\begin{restatable}{theorem}{thmapsisplb}
Given an undirected unweighted graph $G=(V,E)$, computing 2-APSiSP requires $\tilde{\Omega}(n)$ rounds, even if APSP distances are known (i.e. vertex $x$ knows distances $d(x,y)\; \forall y \in V$). This lower bound applies even if $G$ has constant diameter, and against $(3-\epsilon)$-approximation algorithms (for any constant $\epsilon>0$). In directed graphs, it holds for any constant approximation ratio.
    \label{thm:apsisp:lb}
\end{restatable}

\subsection{Prior Work}\label{sec:prior}
\paragraph{Distance Sensitivity Oracles (DSO)}
The DSO problem has been studied extensively in the sequential setting~(e.g. \cite{DemetrescuTCR08dso, BernsteinK08dso, BernsteinK09dso, WeimannY13, Ren22, DeyG24}). No non-trivial algorithms for DSO were known in the distributed CONGEST setting, but some related problems have received attention.

Algorithms and lower bounds for computing replacement path distances for a single pair of vertices (RPaths) were presented in~\cite{ManoharanR24rp}, achieving almost optimal bounds in most cases: These include a $\tilde{\Theta}(n)$ round bound (tight up to polylog factors) for directed weighted graphs, and bounds almost matching $SSSP$ round complexity for undirected graphs; nontrivial upper and lower bounds were also presented for unweighted directed graph. Very recently, matching CONGEST upper and lower bounds for RPaths in directed unweighted graphs were given in~\cite{ChangCDMNS25}.
A distributed $O(D\log n)$ algorithm for single source replacement paths in undirected unweighted graphs was given in~\cite{GhaffariP16fault}. 

Distributed constructions of fault-tolerant preservers and spanners have been studied~\cite{BodwinP23,DinitzR20, Parter22vertex}; these methods construct a sparse subgraph that exactly or approximately preserves replacement path distances. While these constructions could potentially be used during preprocessing, they do not seem to yield an efficient query procedure to answer DSO queries.

Distance oracle algorithms, which preprocess the graph to answer approximate shortest path queries without any edge failures, have received some attention in the distributed setting: a CONGEST algorithm for $(3+\epsilon)$-approximate \textit{undirected} distance oracle was given in~\cite{CHLP21oracle}. There is no prior work on non-trivial distributed distance oracles for directed graphs.

\paragraph{All Pairs Second Simple Shortest Paths (2-APSiSP)} 
The 2-APSiSP problem and 2-SiSP problem are well understood in the sequential setting. A sequential $\tilde{O}(mn)$ time algorithm for 2-APSiSP was given in~\cite{AgarwalR16sisp}, which matches the running time of the best directed weighted 2-SiSP algorithm given in~\cite{Yen1971}. A fine-grained sequential lower bound of $\Omega(mn)$ for directed weighted 2-SiSP was shown in~\cite{AgarwalR18finegrained}, which trivially applies to 2-APSiSP as well. Thus, we have almost tight (conditional) bounds in the sequential setting.

There is no prior work on 2-APSiSP in the distributed setting. Its single-pair variant 2-SiSP has been studied in the \congest model, where tight bounds of $\tilde{\Theta}(n)$ for directed weighted graphs and $\Theta(SSSP)$ for undirected graphs are known~\cite{ManoharanR24rp}.

\subsection{Roadmap} The remainder of this paper is organized as follows. We begin in Section~\ref{sec:exclude} by developing our excluded shortest paths algorithm, a core tool used throughout the paper. In Section~\ref{sec:dso:algs}, we present our DSO algorithms. To complement these algorithms, we present lower bounds in Section~\ref{sec:lbs} for $k$-SSP and DSO. Finally, we present our near-optimal upper and lower bounds for 2-APSiSP in Section~\ref{sec:apsisp}, and conclude with some open problems in Section~\ref{sec:conclusion}.

\section{Distributed Excluded Shortest Paths}
\label{sec:exclude}

\definecolor{mygreen}{HTML}{0ADD08}
\definecolor{myblue}{HTML}{57B9FF}

\begin{figure}[t]
    \centering
    
    \begin{minipage}{0.45\textwidth}
        \scalebox{0.6}{
    \begin{tikzpicture}
        \tikzstyle{vertex}=[circle, draw=black,minimum size=20pt, thick]
    \tikzstyle{vertexsm}=[circle, draw=black,minimum size=5pt, thick]
    \tikzstyle{pedge}=[draw=blue, very thick]
    \tikzstyle{redge}=[draw=mygreen]
    \tikzstyle{gedge}=[draw=red, very thick]

        \node[vertex] (x) at (7,0) {$x$};
        \node[vertex] (r0) at (5,-1) {$a$};
        \node[vertex] (r1) at (7.6,-1) {};
        \node[vertex] (r2) at (9,-1) {};
        \node[vertex] (v1) at (7.4,-2.5) {};
        \node[vertex] (v) at (6.9,-4) {$v$};
        \node (r21) at (10,-3) {};

        \node[vertex] (p1) at (4,-2.5) {$b$};
        \node[vertex] (p2) at (5,-4) {$c$};
        \node at (5.8,-1.7) {\large \textcolor{blue}{${P = \langle a,b,c \rangle}$}};
        \node[vertex] (t1) at (3.2,-4) {};
        \node[vertex] (t11) at (2.2,-5.5) {};
        \node[vertex] (t12) at (3.7,-5.5) {};
        \node[vertex] (z) at (5.4,-5.5) {$z$};

        \node (ur) at (10,-3.2) {\dots};
        \draw[dashed] (1.8,-6) rectangle (5.9,-2);
        \node at (2.3,-1.6) {\large $T_x(P)$};

        \path[draw,thick,->] (x) edge (r0);
        \path[draw,thick,->] (x) edge (r1);
        \path[draw,thick,->] (x) edge (r2);

        \path[draw,thick,->] (r1) edge (v1);
        \path[draw,thick,->] (v1) edge (v);
        \path[draw,thick,->] (r2) edge (r21);

        \path[draw,thick,->,pedge] (r0) edge (p1);
        \path[draw,thick,->,pedge] (p1) edge (p2);

        \path[draw,thick,->] (p1) edge (t1);
        \path[draw,thick,->] (t1) edge (t11);
        \path[draw,thick,->] (t1) edge (t12);
        \path[draw,thick,->] (p2) edge (z);

        \path[draw,dashed,->,thick] (v) edge (z);
        
    \end{tikzpicture}
    }
    \caption{Example graph $G$ with out-shortest path tree $T_x$ for source vertex $x$. Blue path $P=\langle a,b,c \rangle$ denotes an example excluded path. $T_x(P)$ denotes the subtree rooted at key vertex $b$ of $P$ (second vertex on $P$). Solid edges denote tree edges in $T_x$, dashed $v$-$z$ edge denotes non-tree edge in $G$.}
    \label{fig:exclude}
    \end{minipage}%
    \hfill
    \begin{minipage}{0.45\textwidth}
        \scalebox{0.6}{
    \begin{tikzpicture}
        \tikzstyle{vertex}=[circle, draw=black,minimum size=20pt, thick]
    \tikzstyle{vertexsm}=[circle, draw=black,minimum size=5pt, thick]
    \tikzstyle{pedge}=[draw=myblue]
    \tikzstyle{redge}=[draw=mygreen, very thick]
    \tikzstyle{gedge}=[draw=red, very thick]

        \node[vertex] (x) at (7,0) {$x$};
        \node[vertex] (r0) at (5,-1) {$a$};
        \node[vertex] (r1) at (7.6,-1) {};
        \node[vertex] (r2) at (9,-1) {};
        \node[vertex] (v1) at (7.4,-2.5) {};
        \node[vertex] (v) at (6.9,-4) {$v$};
        \node (r21) at (10,-3) {};

        \node[vertex] (p1) at (4,-2.5) {$b$};
        \node[vertex] (p2) at (5,-4) {$c$};
        \node[vertex] (t1) at (3.2,-4) {};
        \node[vertex] (t11) at (2.2,-5.5) {};
        \node[vertex] (t12) at (3.7,-5.5) {};
        \node[vertex] (z) at (5.4,-5.5) {$z$};

        \node (ur) at (10,-3.2) {\dots};
        \draw[dashed] (1.8,-6) rectangle (5.9,-2);
        \node at (2.3,-1.6) {\large $T_x(P)$};

        \path[draw,thick,->] (x) edge (r0);
        \path[draw,thick,->,redge] (x) edge (r1);
        \path[draw,thick,->] (x) edge (r2);

        \path[draw,thick,->,redge] (r1) edge (v1);
        \path[draw,thick,->,redge] (v1) edge (v);
        \path[draw,thick,->] (r2) edge (r21);


        \path[draw,thick,->] (p1) edge (t1);
        \path[draw,thick,->] (t1) edge (t11);
        \path[draw,thick,->] (t1) edge (t12);
        \path[draw,thick,->] (p2) edge (z);

        \path[draw,dashed,->,redge,thick] (v) edge (z);
        \path[draw,dotted,->,gedge,thick] (x) edge[bend left=20] (z);
    \end{tikzpicture}
    }
    \caption{Example graph $G$ where path $P=\langle a,b,c \rangle$ is removed from $T_x$. Green edges denote the shortest $x$ to $z$ path through $v$ when $P$ is excluded. Dotted red edge denotes an added $x$-$z$ edge, assigned the same weight as the green path from $x$ to $z$ in $G$.}
    \label{fig:exclude:removed}
    \end{minipage}
 
\end{figure}

In this section, we present algorithms for the excluded shortest paths problem (see Definition~\ref{def:exclude}). We are given a graph $G=(V,E)$, a set of sources $X \subseteq V$ and an independent set of paths $\mathcal{R}_x$ for each source $x$, and we need to compute $d(x,y,P)$ for each $x \in X, y \in V$ and $P \in \mathcal{R}_x$. 
We first consider the single source case and then extend to multiple sources.

For a single source, we build on a sequential method given in~\cite{DemetrescuTCR08dso}, which introduced this problem. From the single source $x \in V$, consider a path $P \in \mathcal{R}_x$. Recall from Definition~\ref{def:indpaths} that $T_x(P)$ denotes the vertices under the subtree of $T_x$ rooted at the second vertex in $P$ from $x$, which we call the \textit{key} vertex of $P$. We consider edge failures, unlike~\cite{DemetrescuTCR08dso} that considers vertex failures. So paths in $G$ can use the first \textit{vertex} of $P$ when $P$ is excluded and thus the removal of $P$ only affects the distances of vertices in $T_x(P)$. See Figure~\ref{fig:exclude} for an example $T_x(P)$.

Consider a vertex $z$ in $T_x(P)$. We add an edge from source $x$ to vertex $z$ (we call this the $x$-$z$ edge) that captures the minimum distance of a path from $x$ to $z$ that does not use \textit{any} edge in $T_x(P)$. We assign this edge the weight $d^*(x,z) = min_{v \in N_{in}(z), v \not\in T_x(P)} d(x,v)+w(v,z)$~\cite{DemetrescuTCR08dso}, since such a path must reach $z$ through some in-neighbor not in $T_x(P)$. We can now compute the distance $d(x,z,P)$ by computing SSSP from source $x$ in the graph with these $x$-$z$ edges added, and edges in all $P \in \mathcal{R}_x$ removed. This allows us to capture paths that use some edges in $T_x(P)$ but not edges in $P$ itself. See the added $x$-$z$ edge (dotted red edge) in Figure~\ref{fig:exclude:removed} for an example. 

Algorithm~\ref{alg:exclude} implements this approach in the CONGEST setting. In line~\ref{alg:exclude:sssp}, we compute SSSP distances from source $x$. In line~\ref{alg:exclude:edgedown}, we send the identity of the key vertex of $P$ from $x$ to each vertex in subtree $T_x(P)$ using a downcast along $T_x$. This allows each vertex to know which subtree $T_x(P)$ it belongs to. Note that since $\mathcal{R}_x$ is independent, any vertex $y$ belongs to at most one $T_x(P)$. The downcast takes $O(n)$ rounds and $O(1)$ congestion for source $x$. 

Each vertex $y$ in $T_x(P)$ sends distance $d(x,y)$ and the key vertex of $P$ to all its neighbors in line~\ref{alg:exclude:neighb} in $O(1)$ rounds. Using this neighbor information, we compute the weight $d^*(x,z)$ of the added $x$-$z$ edge locally at $z$ in line~\ref{alg:exclude:edgewt}. Finally, the SSSP computation in line~\ref{alg:exclude:finalsssp} computes all excluded shortest paths distances, and the distance $d(x,z,P)$ is known at vertex $z$. Correctness of this procedure follows from~\cite{DemetrescuTCR08dso}.

Note that all steps other than the SSSP computations in 
lines~\ref{alg:exclude:sssp} and~\ref{alg:exclude:finalsssp} take $O(n)$ rounds and have $O(1)$ congestion. To maintain this low congestion, we use the directed weighted SSSP algorithm of~\cite{GhaffariT24sssp} which takes $\tilde{O}(n)$ rounds and has $\tilde{O}(1)$ congestion (instead of, say,~\cite{CaoF23parallel} with high congestion). Our total round complexity is $\tilde{O}(n)$ rounds and $\tilde{O}(1)$ congestion. 

Although we focus on algorithms for directed weighted graphs, we note that if the graph is directed and unweighted we can remove the polylog factors to obtain a bound of $O(n)$ rounds and $O(1)$ congestion. For this we use $O(n)$-round BFS to compute shortest path distances in line~\ref{alg:exclude:sssp}, and simulate the $x$-$z$ weighted edges (which have weight $d^*(x,z) < n$) added in line~\ref{alg:exclude:finalsssp} locally at $z$ by assuming a message from $x$ arrives at round $d^*(x,z)$ of the BFS.

\begin{algorithm}[t]
    \caption{Single Source Excluded Shortest Paths}
    \begin{algorithmic}[1]
        \Require Graph $G=(V,E)$, source $x \in V$, independent set of paths $\mathcal{R}_x$.
        \Ensure Compute distance $d(x,y,P)$ for each $y \in V, P \in \mathcal{R}_x$.
        \State Compute SSSP from $x$ to compute $d(x,y)$ at $y$ for each $y \in V$.\label{alg:exclude:sssp}
        \State Using a downcast from $x$ along $T_x$, communicate key vertex of $P \in \mathcal{R}_x$ to each $z \in T_x(P)$. Vertex $z \in V$ now knows which subtree $T_x(P)$ it is in, if any.\label{alg:exclude:edgedown}
        \State Each vertex $y$ communicates $d(x,y)$ to its neighbors $N_{out}(y)$. If $y$ is in a subtree $T_x(P)$ (as identified in line~\ref{alg:exclude:edgedown}), it communicates the key vertex of $P$ as well to its neighbors. \label{alg:exclude:neighb}
        \State Each vertex $z \in T_x(P)$ locally computes $d^*(x,z) = \min\limits_{v \in N_{in}(z), v \not\in T_x(P)} (d(x,v)+w(v,z))$. The check for $v \not\in T_x(P)$ is done locally using the information from line~\ref{alg:exclude:neighb}. \label{alg:exclude:edgewt}
        \State Perform SSSP from $x$ on $G$ with edges in $\mathcal{R}_x$ removed, and added $x$-$z$ edge with weight $d^*(x,z)$. These added edges are simulated locally at $z$ in the course of the SSSP algorithm. \label{alg:exclude:finalsssp}
    \end{algorithmic}
    \label{alg:exclude}
\end{algorithm}

\subsection{Multiple-Source Excluded Shortest Paths.}
Algorithm~\ref{alg:exclude} gives us an $\tilde{O}(1)$-congestion algorithm for computing excluded shortest paths from one source. To compute excluded shortest paths from multiple sources, we use random scheduling~(see Section~\ref{sec:primitives}). This proves Theorem~\ref{thm:dso:exclude} restated below. We will use this result in our first DSO algorithm and our 2-APSiSP algorithm. 

\thmdsoexclude*

\section{Algorithms for Distance Sensitivity Oracles}
\label{sec:dso:algs}

We now present our DSO algorithms for directed weighted graphs. They also trivially apply to undirected or unweighted graphs. We first consider answering DSO queries with no preprocessing in Section~\ref{sec:dso:nopre}.
We then present two algorithms that utilize preprocessing to improve query rounds in Sections~\ref{sec:dso:algquery} and~\ref{sec:dso:algpre}.

In our DSO algorithms we assume that edge $e$ is on an $x$-$y$ shortest path. If APSP distances have been computed, we can readily check if $e=(z,z')$ is on a $x$-$y$ shortest path by checking if $d(x,y) = d(x,z) + w(e)+d(z',y)$. If this condition does not hold, we can immediately answer $d(x,y,e)=d(x,y)$.

\subsection{DSO and seb-DSO without Preprocessing}
\label{sec:dso:nopre}
Without any preprocessing, a batch of $k$ arbitrary queries can be trivially answered in $k \cdot SSSP$ rounds (the current best SSSP algorithm takes $\tilde{O}(\sqrt{n}+n^{2/5+o(1)}D^{2/5}+D)$~\cite{CaoF23parallel} rounds). 
This works even for general batched DSO, where the failed edge may vary across queries. 

For seb-DSO, where all queries concern the same failed edge, we can improve on $k \cdot SSSP$ rounds without any preprocessing. We show below that we can achieve the same upper bounds as $k$-SSP, establishing Observation~\ref{thm:dso:seb} from Section~\ref{sec:dso:sebintro} for directed graphs. Similarly for undirected graphs, we can use the previously known $k$-SSP algorithms in~\cite{ElkinN19} ($(1+\epsilon)$-approximate undirected weighted $k$-SSP) and in~\cite{LenzenPP19,HoangPDGYPR19} (exact undirected unweighted $k$-SSP).

In Corollary~\ref{lem:seblb}~(Section~\ref{sec:dso:lb}), we prove a $\tilde{\Omega}(\sqrt{nk}+D)$ lower bound for seb-DSO without preprocessing in directed graphs and undirected weighted graphs. This is derived from the $k$-SSP lower bound in Theorem~\ref{thm:kssp:lb}~(Section~\ref{sec:results:lb}): For directed graphs, our bound is tight for $k \ge n^{1/3}$, and in undirected weighted graphs, our bound is almost optimal for all $k$ for approximate seb-DSO.

\begin{proof}[Proof of Observation~\ref{thm:dso:seb}]
    Let the set of $k$ queries be $d(x_i,y_i,e)$ for $1 \le i \le k$, with $e$ being the single failed edge. Among these sources $U=\{x_1,x_2,\dots x_k\}$, if there are $|U| = k' \le k$ unique vertices, we compute $k'$-SSP in the graph $G-\{e\}$ with source set $U$. This computes each distance $d(x_i,y_i,e)$ at $y_i$, using a $k$-SSP algorithm.
\end{proof}

If all queries of the seb-DSO have the same source or sink
we can do even better:
we only need one SSSP computation with the failed edge removed (in the reversed graph in the case of a single sink), and we can achieve $O(SSSP)$ query cost.

\subsection{DSO with Fast Query Responses}
\label{sec:dso:algquery}
In this section, we present a CONGEST method for directed weighted DSO that has fast query responses, and works for general batched DSO. This algorithm takes $\tilde{O}(n^{3/2})$ rounds to preprocess the graph, has oracle size $\tilde{O}(n)$ per node (for a total of $\tilde{O}(n^2)$ size), and answers $k$ arbitrary queries in $O(k+D)$ rounds.

Our method, described in Algorithm~\ref{alg:dsoquery}, builds on a sequential DSO algorithm in~\cite{BernsteinK08dso} that runs in time $\tilde{O}(mn^{1.5})$ and answers each query in $O(1)$ time. This sequential construction achieves an oracle size of $\tilde{O}(n^2)$, improving on a previous sequential algorithm of~\cite{DemetrescuTCR08dso} that has the same preprocessing time but an oracle size $\tilde{O}(n^{2.5})$. As mentioned earlier, implementing~\cite{DemetrescuTCR08dso} in \congest incurs an oracle size of $\tilde{O}(n^{3/2})$ per node, a factor of $\sqrt{n}$ larger than the $\tilde{O}(n)$ size per node we achieve.

The sequential algorithm of~\cite{BernsteinK08dso} achieves reduced oracle size by using randomized sampling to partition shortest paths into intervals of geometrically increasing lengths, instead of the $\sqrt{n}$-length intervals in~\cite{DemetrescuTCR08dso}. This allows for a better reuse of intervals in different queries, leading to reduced oracle size. However, 
a distributed implementation of this algorithm requires communicating information about the sampled vertices and intervals in the graph, which is typically expensive in the \congest model. We show how to implement this efficiently, while maintaining $\tilde{O}(n^{3/2})$ preprocessing rounds and achieving $\tilde{O}(n)$ oracle size per node.
 
We begin by reviewing elements of the sequential algorithm in~\cite{BernsteinK08dso} in Section~\ref{sec:dso:center}. Next, in Section~\ref{sec:dso:centerinterval} we show how to implement these ideas in \congest while avoiding high communication costs. In Section~\ref{sec:dso:excludeinterval}, we describe how we compute various types of excluded shortest path distances using different methods during the preprocessing phase. Finally, we combine these components in Algorithm~\ref{alg:dsoquery} and provide its analysis in Section~\ref{sec:dso:dsoqueryana}.

\subsubsection{Centers and Intervals} 
\label{sec:dso:center}
A key element of~\cite{BernsteinK08dso} is to randomly sample vertices to be \textit{centers} of varying \textit{priority}. For $1 \le j \le \log n$, the algorithm samples $\tilde{\Theta}(n/2^j)$ vertices to be centers of priority $j$. Every vertex is at least a priority 1 center
and each vertex is assigned the largest priority for which it is sampled. 

Given an $x$-$y$ shortest path, we identify the sequence of centers of strictly increasing priority starting from $x$. Once the maximum priority center is reached, we continue with centers of strictly decreasing priorities till we reach $y$. This creates a \textit{center sequence} of at most $2\log n$ vertices, that partitions the shortest path into \textit{center intervals}.

A center $v$ in this sequence is said to \textit{cover} the edges in the intervals with $v$ as an endpoint, and $d(x,v,e)$ or $d(v,y,e)$ is computed for each $e$ in the interval (depending on whether $e$ is on the $x$-$v$ or $v$-$y$ segment of the shortest path respectively). These distances are computed using excluded shortest paths computations, by excluding all edges at a fixed hop-length from $v$ at the same time. A center of priority $j$ only has to cover edges up to $\tilde{O}(2^j)$ depth \whp in $n$, hence we perform $\tilde{O}(2^j)$ exclude computations
for a priority $j$ center.

We denote an interval with endpoints $c_x, c_y$ as $(c_x,c_y)$. If we know the distance $d(x,y,(c_x,c_y))$, i.e., the distance when all edges in the $(c_x,c_y)$ interval are excluded, then, from~\cite{DemetrescuTCR08dso} we can compute the replacement path distance as:

\begin{equation*}
    d(x,y,e) = \min(d(x,c_x)+d(c_x,y,e), d(x,c_y,e)+d(c_y,y), d(x,y,(c_x,c_y)))
\end{equation*}

The first term in the above expression captures replacement paths that diverge from the $x$-$y$ shortest path after $c_x$, and the second captures paths that rejoin before $c_y$. The remaining type of path diverges before $c_x$ and rejoins after $c_y$, and thus avoids the whole $(c_x,c_y)$ interval. The third term, $d(x,y,(c_x,c_y))$, handles paths that diverge before $c_x$ and rejoin after $c_y$, avoiding the interval entirely. 

It may not be feasible to compute $d(x,y,(c_x,c_y))$ efficiently for all intervals $(c_x,c_y)$.
Instead, if we define $D(x,y,(c_x,c_y)) = \max_{e \in (c_x,c_y)} d(x,y,e)$, it is shown in (\cite{BernsteinK08dso}, Lemma 4.1) that $d(x,y,e)$ is correctly computed by Equation~\ref{eqn:dso} below, which uses $D(x,y,(c_x,c_y))$ in place of $d(x,y,(c_x,c_y))$.

\begin{align}\label{eqn:dso}
    d(x,y,e) & = \min(d(x,c_x)+d(c_x,y,e), d(x,c_y,e)+d(c_y,y), D(x,y,(c_x,c_y)))  \notag \\
    & \; \text{where } D(x,y,(c_x,c_y)) = \max_{e \in (c_x,c_y)} d(x,y,e)
\end{align}

Note that it may be the case that $d(x,y,(c_x,c_y))\ne D(x,y,(c_x,c_y))$.
For convenience, in Algorithm~\ref{alg:dsoquery}, we will use the notation $D(x,y,(c_x,c_y))$ throughout even in cases where we compute $d(x,y,(c_x,c_y))$ directly (e.g. line~\ref{alg:dsoquery:longcompute}).

\subsubsection{Computing Center Intervals} 
\label{sec:dso:centerinterval}
In the sequential setting, computing center information along each $x$-$y$ shortest path for $x,y \in V$ can be performed readily using APSP and BFS computations. However, in the \congest setting, care is required to make relevant center information available at each node while avoiding congestion and increased oracle size. For instance, at a single node $v$, we cannot store the identities of every interval it is on as this could incur $\Omega(n^2)$ space at a node that is on many shortest paths; this problem does not arise in the centralized setting where the total size would still be $\tilde{O}(n^2)$. Additionally, the intervals for an $x$-$y$ path and an $x$-$y'$ path may be different even if $y$ is on the path from $x$ to $y'$: it depends on the priorities of intermediate vertices. 

In order to resolve these issues in the \congest setting, we first augment a \congest APSP algorithm, such as the one from~\cite{BernsteinN19apsp} to also pass on parent vertex, hop length and center information. The center information for $v \in V$ is the center closest to $v$ on the $x$-$v$ path with priority $j$, denoted $\CC{x}{v}{j}$ for each $x \in V$ and for each $1 \le j \le \log n$. The APSP algorithm can be readily modified to do so: In the base case of a single edge, this information is trivially computed. Whenever the algorithm combines two distances $d(x,u)$ and $d(u,y)$ to form distance $d(x,y)$, we update the center information for the $x$-$y$ path from the information for $x$-$u$ and $u$-$y$ paths. Note that the APSP algorithm would compute this information at the target vertex $y$. To make it available to $x$ as well, we repeat this on the reversed graph $G^r$. See Algorithm~\ref{alg:dsoapsp} for details.
Algorithm~\ref{alg:dsoapsp} is called in line~\ref{alg:dsoquery:apsp} of Algorithm~\ref{alg:dsoquery}.

Using the computed center information $\CC{x}{v}{j}$, we show how to locally compute the endpoints of a center interval in Algorithm~\ref{alg:dsointerval}. Note that this is  readily done in the sequential setting, but requires care in CONGEST since relevant values are stored at different nodes. This is called as a subroutine in Algorithm~\ref{alg:dsoquery} in several places to compute intervals whenever needed. We store the center information $\CC{x}{v}{j}, \CC{v}{x}{j}$ at a vertex $v$, which is $O(\log n)$ size per vertex $x$, and don't explicitly store the actual intervals for each $x$-$y$ path that $v$ may be on. In addition to node $v$, center information for certain \textit{covered} edges is also stored at endpoint $y$ of relevant paths (see line~\ref{alg:dsoquery:mainexclude} of Algorithm~\ref{alg:dsoquery}). This enables local computation of intervals using Algorithm~\ref{alg:dsointerval} when needed. By bounding the number of such covered edge and vertex pairs, we ensure that it does not take up too much space. During query time as well, we recompute interval information from the stored center information (see line~\ref{alg:dsoquery:queryinterval} of Algorithm~\ref{alg:dsoquery}). Thus, we maintain $\tilde{O}(n)$ oracle size at each node.

\subsubsection{Computing Excluded Interval Distances} 
\label{sec:dso:excludeinterval}
As noted earlier in Section~\ref{sec:dso:center}, we need to compute $x$-$y$ shortest path distances when all edges in a center interval $(c_x,c_y)$ are removed. Equivalently, using (Lemma 4.1~\cite{BernsteinK08dso}), we can replace this with the maximum edge replacement distance, $D(x,y,(c_x,c_y)) = \max_{e \in (c_x,c_y)} d(x,y,e)$. For ease of notation, we use $D(x,y,(c_x,c_y))$ even when we compute the exact excluded interval distance.

We compute the excluded interval distance differently depending on whether or not $c_x$ or $c_y$ has priority $\ge \frac{1}{2}\log n$ (similar to~\cite{BernsteinK08dso}). If an endpoint of the interval has priority $\ge \frac{1}{2}\log n$, we call this a long interval. Otherwise, \whp in $n$, the interval is within $\tilde{O}(\sqrt{n})$ hops of $x$ or $y$ and we call it a short interval.

We collect all long intervals in the shortest path tree rooted at $x$, and compute excluded shortest paths from $x$ for each interval one at a time (see lines~\ref{alg:dsoquery:longbroad},\ref{alg:dsoquery:longcompute} of Algorithm~\ref{alg:dsoquery}). For a fixed $x$, there are $\tilde{O}(\sqrt{n})$ intervals with at least one endpoint of priority $\ge \frac{1}{2}\log n$. Adding up over all $x \in V$, we get a total of $\tilde{O}(n^{3/2})$ intervals. Thus, we broadcast $\tilde{O}(n^{3/2})$ messages of the form $L(x,y,(c_x,c_y))$ in line~\ref{alg:dsoquery:longbroad}, and line~\ref{alg:dsoquery:longcompute} makes $\tilde{O}(n^{3/2})$ exclude computations. Using Theorem~\ref{thm:dso:exclude}, this takes $\tilde{O}(n^{3/2})$ rounds.

In the other case when both endpoints of a $x$-$y$ path interval have priority $< \frac{1}{2}\log n$, this short interval is within $\Theta(\sqrt{n})$ hops of $x$ or $y$ \whp in $n$. For each vertex $x \in V$, we compute excluded shortest paths distances from $x$ when each edge within $\Theta(\sqrt{n})$ hops is excluded line~\ref{alg:dsoquery:mainexclude}.
We use these exclude computations to also forward center information, so that we can locally determine the entire short interval at $y$. Finally, to compute distances for excluded short intervals, we locally compute $D(x,y,(c_x,c_y)) = \max_{e \in (c_x,c_y)} d(x,y,e)$ at $y$ in line~\ref{alg:dsoquery:shortcompute}.

\textbf{Why not use bottleneck vertices as in~\cite{BernsteinK09dso}?} The sequential algorithm of \cite{BernsteinK09dso} achieves $\tilde{O}(mn)$ preprocessing time by changing how excluded interval distances are computed. A bottleneck vertex $BV(x,y,(c_x,c_y))$ is identified in each interval $(c_x,c_y)$ that maximizes the replacement distance. Then, $D(x,y,(c_x,c_y))$ is computed by excluding just this bottleneck vertex. This computation is done for all intervals at once using a single SSSP on a constructed $\tilde{O}(n^2)$ vertex graph.
The second step can be implemented in \congest in $\tilde{O}(n)$ rounds by breaking up the $\tilde{O}(n^2)$ vertex graph into smaller $\tilde{O}(n)$ vertex graphs. However, the first step of finding the bottleneck vertices utilizes a $O(\log n)$ binary search procedure that needs distances stored at different nodes. Implementing this in \congest seems to require $\tilde{O}(n^2)$ broadcasts, blowing up the preprocessing time.

\subsubsection{Algorithm~\ref{alg:dsoquery} and Analysis} 
\label{sec:dso:dsoqueryana}

Algorithm~\ref{alg:dsoquery} works as follows: In the preprocessing algorithm, we first sample centers in line~\ref{alg:dsoquery:sample}. We compute shortest paths along with parent vertex, hop length, center information in line~\ref{alg:dsoquery:apsp} by calling the \textsc{ModifiedAPSP} subroutine described in Algorithm~\ref{alg:dsoapsp}. Then, in line~\ref{alg:dsoquery:mainexclude}, we compute excluded shortest paths distances for single edge removal. These are of two types: From each center excluding each edge it covers, and from each vertex upto $\sqrt{n}$ hops in order to compute short interval excluded distances locally. 

Now, we compute excluded shortest paths distances for entire intervals. This is again of two types: In lines~\ref{alg:dsoquery:longtrst}-\ref{alg:dsoquery:longtren}, we identify all long intervals and broadcast them. In line~\ref{alg:dsoquery:longcompute}, we compute the actual excluded distances for long intervals by collecting the broadcast messages. In lines~\ref{alg:dsoquery:shorttrst}-\ref{alg:dsoquery:shorttren}, we locally compute excluded distances for short intervals, using the distances computed in line~\ref{alg:dsoquery:mainexclude}. 

The final step of preprocessing in lines~\ref{alg:dsoquery:cleanst}-\ref{alg:dsoquery:cleanen} is to remove unnecessary information stored at nodes to achieve $\tilde{O}(n)$ oracle size per node. 

\begin{algorithm}[htp]
    \caption{DSO with fast query: $\tilde{O}(n^{3/2})$ rounds preprocessing, $O(k+D)$ rounds for $k$ queries and $\tilde{O}(n)$ oracle size per node.}
    \begin{algorithmic}[1] 
        \Require Directed weighted graph $G=(V,E)$.
        \Ensure Compute answers to $k$ replacement path distance queries $d(x_i,y_i,e_i)$, for $1 \le i \le k$ and $x_i,y_i \in V, e_i \in E$, that are broadcast to all nodes in $G$. $d(x_i,y_i,e_i)$ must be known at $y_i$ upon termination.
        \lineComment{\underline{Preprocessing Algorithm}}
        \ForAll{$v \in V$}
            \State Vertex $v$ randomly chooses a \textit{priority} $p_v$ for itself from $\{1,2, \dots \lceil \log n \rceil\}$, with a probability distribution choosing priority $j$ with probability $c/2^j$, for appropriate constant $c$. \label{alg:dsoquery:sample}\Comment{Local Computation at $v$}
        \EndFor
        \State Run \textsc{ModifiedAPSP}($G$) (see Algorithm~\ref{alg:dsoapsp}). For each $x,y \in V$, shortest path distance $d(x,y)$, hop length $h(x,y)$, parent vertex $p(x,y)$ of $y$,  the closest priority $j$ vertex to $y$, denoted $\CC{x}{y}{j}$ for $1 \le j \le \log n$ are computed at $x$ and $y$. \label{alg:dsoquery:apsp}
        \ForAll{$v \in V$}
            \State If $v$ has priority $p_v$, perform $\max(\Theta(2^{p_v}\log n), \sqrt{n})$ excluded shortest paths computations from $v$. In the $i$'th exclude computation, the set of excluded edges is all edges at depth $i$ in $T_v$. For edge $e=(u,u')$, $u$ can determine its depth using $h(v,u)$ from line~\ref{alg:dsoquery:apsp}. Along with the exclude computation, $u$ also sends $P(v,u)$ and $\CC{v}{u}{j}$, $1 \le j \le \log n$ to each $y \in V$ in the subtree rooted at $u'$ in $T_x$. Repeat this in the reversed graph $G^r$.
            \label{alg:dsoquery:mainexclude}
            \If{priority of $v$ is $p_v \ge \frac{1}{2}\log n$}  \label{alg:dsoquery:longtrst}
                \lineComment{Lines~\ref{alg:dsoquery:longbroadst}-\ref{alg:dsoquery:longbroaden}: At $v$, locally identify long center intervals with endpoint $v$, and broadcast a message for each such interval.}
                \For{edge $e = (v,v')$ incident to $v$} \label{alg:dsoquery:longbroadst}
                    \lineComment{At $v$, locally find all pairs $x,y \in V$ for which $e$ is first edge of an increasing center interval of $x$-$y$ shortest path.}
                    \For{$x,y \in V$}
                        \State Let $(c_x,c_y) \gets \textsc{FindInterval}(x,y,e)$ \Comment{$c_x,c_y$ are endpoints of interval containing $e$. Local computation at $v$, see Algorithm~\ref{alg:dsointerval}.}
                        \If{$c_x = v$ and $p_{c_y} > p_{c_x}$} \Comment{Increasing long interval.}
                            \State Construct a message $L(c_x,c_y,x,y)$ indicating this is a long interval, and broadcast. \Comment{This is used in line~\ref{alg:dsoquery:longcompute}.} \label{alg:dsoquery:longbroad}
                        \EndIf
                    \EndFor 
                \EndFor
                \State Repeat lines~\ref{alg:dsoquery:longbroadst}-\ref{alg:dsoquery:longbroad} for edges $e=(v',v)$. \Comment{Locally finds and broadcasts decreasing center intervals where $e$ is the last edge. }  \label{alg:dsoquery:longbroaden}
            \EndIf \label{alg:dsoquery:longtren}
        \EndFor 
        \ForAll{$x \in V$} \Comment{Compute distances when long intervals are excluded.}
            \State Let $x$ receive $f$ messages of the form $L(c_x,c_y,x,y)$ from line~\ref{alg:dsoquery:longbroad}. Perform $f$ exclude computations from $x$, where in each computation all edges in a {center interval} $(c_x,c_y)$ are excluded. After the broadcast in line~\ref{alg:dsoquery:longbroad}, using \textsc{FindInterval} each edge can locally determine if it is excluded. Vertex $y$ now knows $D(x,y,(c_x,c_y))$ for each long interval $(c_x,c_y)$. \label{alg:dsoquery:longcompute}
        \EndFor
        \Statex Continued on the next page \dots
        \algstore{break:dsoqueryalg}
    \end{algorithmic}
    \label{alg:dsoquery}
\end{algorithm}

\begin{algorithm}[htp]
    \begin{algorithmic}[1]

        \algrestore{break:dsoqueryalg}
        
        \ForAll{$y \in V$} \Comment{Locally at $y$, compute excluded short interval distances.} \label{alg:dsoquery:shorttrst}
            \For{$x \in V$} \Comment{Exclude short intervals in $x$-$y$ shortest path.} \label{alg:dsoquery:shortst}
                \For{$e = (v,v') \in E$ for which $y$ received distance $d(x,y,e)$ in line~\ref{alg:dsoquery:mainexclude}}
                    \State $(c_x, c_y) \gets \textsc{FindInterval}(x,y,e)$. \Comment{Local computation at $y$. $y$~knows $\CC{v}{y}{j}$ and received $\CC{x}{v}{j}$ in line~\ref{alg:dsoquery:mainexclude}, see Algorithm~\ref{alg:dsointerval}.}
                    \If{$h(x,v) \le \sqrt{n}$ and $c_y = v'$} \Comment{Locally at $y$, compute short interval excluded distance if $e$ is the last edge on $(c_x, c_y)$}
                        \State Using parent vertex values $p(x,v)$ received by $y$ in line~\ref{alg:dsoquery:mainexclude}, locally find all edges $e'$ in the interval $c_x, c_y$ by tracing the shortest path from $x$ until $c_x$ is reached. Compute $D(x,y,(c_x,c_y)) = \max_{e' \in (c_x,c_y)} d(x,y,e')$ and store it at $y$. \label{alg:dsoquery:shortcompute}
                    \EndIf
                \EndFor
            \EndFor
            \State Repeat lines~\ref{alg:dsoquery:shortst}-\ref{alg:dsoquery:shortcompute} in the reversed graph.
        \EndFor \label{alg:dsoquery:shorttren}
        \ForAll{$y \in V$} \Comment{Locally cleanup storage at $y$ to achieve $\tilde{O}(n)$ oracle size.} \label{alg:dsoquery:cleanst}
        \State Remove stored values $p(x,v)$, $\CC{x}{v}{j}$ values stored at $y$ from line~\ref{alg:dsoquery:mainexclude}. Remove $d(v,y,e)$ received from line~\ref{alg:dsoquery:mainexclude} if $h(v,y) > \Theta(2^{p_v} \log n)$.
        \State In the final oracle, the only values stored at $y$ are: $d(x,y), p(x,y), \CC{x}{y}{j}, \CC{y}{x}{j}$ for all $x \in V, 1\le j \le \log n$ from line~\ref{alg:dsoquery:apsp}, remaining $d(v,y,e)$ values from line~\ref{alg:dsoquery:mainexclude}, $D(x,y,(c_x,c_y))$ values from lines~\ref{alg:dsoquery:longcompute},\ref{alg:dsoquery:shortcompute}. \label{alg:dsoquery:storage}
        \EndFor \label{alg:dsoquery:cleanen}
        \lineComment{\underline{Query Algorithm}: Queries $d(x_i,y_i,e_i)$ for $1 \le i \le k$ are broadcast to all nodes} 
        \lineComment{All broadcasts in the following lines can be pipelined together, for $1 \le i \le k$ to run in $O(k+D)$ rounds.}
        \For{$1 \le i \le k$}
        \State Let $e_i=(v_i, v_i')$. 
        \State Vertex $v_i$ locally computes $(c_{x_i}, c_{y_i}) \gets \textsc{FindInterval}(x_i,y_i,e_i)$, and broadcasts $c_{x_i}, c_{y_i}$. \label{alg:dsoquery:queryinterval} \Comment{Local computation at $v_i$, see Algorithm~\ref{alg:dsointerval}}
        \State Vertex $x_i$ broadcasts $d(x_i,c_{x_i}), d(x_i,c_{y_i},e_i)$. \label{alg:dsoquery:querybc}\Comment{Computed in line~\ref{alg:dsoquery:mainexclude}}
        \State Compute locally at $y_i$ (Equation~(\ref{eqn:dso})): 
        { \setlength{\abovedisplayskip}{0pt}
\setlength{\belowdisplayskip}{0pt}
\setlength{\abovedisplayshortskip}{0pt}
\setlength{\belowdisplayshortskip}{0pt}
        \begin{align*}
            d(x_i,y_i,e_i) = \min\big(& d(x_i,c_{x_i})+d(c_{x_i},y_i,e_i), \\  
            & \; d(x_i,c_{y_i},e_i)+d(c_{y_i},y_i), \\
            & \; D(x_i, y_i, (c_{x_i},c_{y_i}))\big)
        \end{align*}
        }
        \Comment{$d(c_{y_i},y_i)$ is computed at $y_i$ in line~\ref{alg:dsoquery:apsp}, $d(c_{x_i},y_i,e_i)$ is computed at $y_i$ in line~\ref{alg:dsoquery:mainexclude}, $D(x_i, y_i, (c_{x_i},c_{y_i}))$ is computed at $y_i$ in line~\ref{alg:dsoquery:longcompute} or~\ref{alg:dsoquery:shortcompute}.} \label{alg:dsoquery:queryfinal}
        \EndFor
    \end{algorithmic}
\end{algorithm}

\begin{algorithm}[htp]
    \caption{Modified APSP in line~\ref{alg:dsoquery:apsp} of Algorithm~\ref{alg:dsoquery}}
    \begin{algorithmic}[1]
        \Procedure{ModifiedAPSP}{$G$}
            \State Compute APSP on $G$, using the algorithm of~\cite{BernsteinN19apsp}, computing shortest path distances $d(x,y)$ for $x,y \in V$ along with the following:
            \begin{enumerate}[topsep=0cm,itemsep=0.05cm]
                \item hop length $h(x,y)$ ,i.e. number of edges on the $x$-$y$ path
                \item parent vertex $p(x,y)$ that is closest to $y$
                \item the closest priority $j$ vertex to $y$, $\CC{x}{y}{j}$ on the $x$-$y$ shortest path, for each $1 \le j \le \log n$.
            \end{enumerate} 
            \Comment{These values are computed at $y$.}
            \State To track this information, whenever the algorithm of~\cite{BernsteinN19apsp} combines distances $d(x,v)$ and $d(v,y)$ to form distance $d(x,y)$, we also perform the update:
            \begin{enumerate}[topsep=0cm,itemsep=0.05cm]
                \item $h(x,y) \gets h(x,v) + h(v,y)$ 
                \item $p(x,y) \gets p(v,y)$
                \item for each $1 \le j \le \log n$, if $\CC{v}{y}{j} \ne \nil, \CC{x}{y}{j} \gets \CC{v}{y}{j}$; else $\CC{x}{y}{j} \gets \CC{x}{v}{j}$
            \end{enumerate} 
            \State Repeat on the reversed graph $G^r$ to make $\CC{x}{y}{j}$, $h(x,y)$, $p(x,y)$ available at node $x$, 
        \EndProcedure
    \end{algorithmic}
    \label{alg:dsoapsp}
\end{algorithm}

\begin{algorithm}[htp]
    \caption{\textit{Local Computation}: Helper Function for Algorithm~\ref{alg:dsoquery}}
    \begin{algorithmic}[1]
        \Require $x,y \in V, e = (v,v') \in E$. Requires local knowledge of $\CC{x}{v}{j}, \CC{v}{y}{j}$ for $1 \le j \le \log n$, these values are computed in Algorithm~\ref{alg:dsoapsp}.
        \Ensure Locally compute endpoints of the center interval of $x$-$y$ shortest path containing $e$. This local computation is at $v$ or $y$, depending on context. In case it is at $y$, $\CC{x}{v}{j}$ values are sent to $y$ before this computation.
        \Procedure{FindInterval}{$x,y,e$}
            \State Let $j_x$ be the maximum $j$ for which $\CC{x}{v}{j} \ne \nil$. \label{alg:dsointerval:st} \Comment{Highest priority on $x$-$v$ path} 
            \If{$\CC{v}{y}{j_x+1} \ne \nil$}
                \State \Return (\CC{x}{v}{j_x}, \CC{v}{y}{j_x+1}) \Comment{$e$ is in an increasing interval}
            \Else 
                \lineComment{Lines \ref{alg:dsointerval:decst}-\ref{alg:dsointerval:decen}: $e$ is in a decreasing or highest priority interval}
                \State Let $j_y$ be the maximum $j$ for which $\CC{v}{y}{j} \ne \nil$. \Comment{Highest priority on $v$-$y$ path. Note that $j_y < j_x+1$} \label{alg:dsointerval:decst}
                \If{$j_y=j_x$} \Comment{Both have equal highest priority}
                    \State \Return (\CC{x}{v}{j_x}, \CC{v}{y}{j_y}) 
                \Else \Comment{Decreasing interval}
                    \State \Return (\CC{x}{v}{j_y+1}, \CC{v}{y}{j_y}) \label{alg:dsointerval:decen}
                \EndIf
            \EndIf
        \EndProcedure
    \end{algorithmic}
    \label{alg:dsointerval}
\end{algorithm}

Now that all excluded distances are precomputed, the query algorithm is straightforward. We identify the center interval using a broadcast in line~\ref{alg:dsoquery:queryinterval}. Based on this interval, we broadcast the relevant distance information stored at different nodes~\ref{alg:dsoquery:querybc}. Using this information, the final answer to query $d(x_i,y_i,e_i)$ is computed at $y_i$ in line~\ref{alg:dsoquery:queryfinal}.

We now prove that Algorithm~\ref{alg:dsoquery} constructs the DSO in Theorem~\ref{thm:dso:algquery}, restated below.

\thmdsoalgquery*
\begin{proof}
    \textbf{Correctness:} For ease of explanation, consider a single query $d(x,y,e)$.
    Let $P_{xy}$ be a replacement path from $x$ to $y$ when $e$ is removed. Let $c_x, c_y$ be the endpoints of the {center interval} of the $P_{xy}$ that contains $e$ (see Section~\ref{sec:dso:center}). Once the excluded shortest path distances are computed in lines~\ref{alg:dsoquery:longcompute} and \ref{alg:dsoquery:shortcompute}, the correctness query algorithm follows from Equation~\ref{eqn:dso} and (Lemma 4.1~\cite{BernsteinK08dso}): If the replacement passes through $c_x$ or $c_y$, the expression $d(x,c_x)+d(c_x,y,e)$ or $d(x,c_y,e)+d(c_y,y)$ respectively captures the correct distance. If it doesn't pass through $c_x$ or $c_y$, $D(x, y, (c_x,c_y))$ correctly captures the distance.

    \textbf{Round Complexity:} Computing modified APSP in line~\ref{alg:dsoquery:apsp} takes $\tilde{O}(n)$ rounds.
    
    In line~\ref{alg:dsoquery:mainexclude}, the algorithm computes distances when excluding edges one at a time.
    We analyze the computations differently depending on whether $p_v \ge \frac{1}{2}\log n$. If $p_v < \frac{1}{2}\log n$, we perform $\sqrt{n}$ exclude computations per vertex $v$, which amounts to at most $\tilde{O}(n^{3/2})$ exclude computations over all $v \in V$. If $p_v \ge \frac{1}{2}\log n$, we perform $\tilde{\Theta}(2^{p_v})$ excludes with source $v$. Due to our sampling probability, there are $\tilde{\Theta}(n/2^{p_v})$ vertices of priority $p_v$ \whp in $n$, leading to a total of $\tilde{O}(n)$ exclude computations for a given priority value $p_v$. Adding up the costs over all choices of $\frac{1}{2}\log n \le p_v \le \log n$, we get a total of $\tilde{O}(n)$ exclude computations. Thus, in line~\ref{alg:dsoquery:mainexclude}, we perform a total of $\tilde{O}(n^{3/2})$ exclude computations.

    Now, we analyze the distance computations where whole intervals are excluded.
    In line~\ref{alg:dsoquery:longbroad}, we compute all long intervals, where at least one endpoint has priority $p_v \ge \frac{1}{2}\log n$. As noted in (Lemma 8.1~\cite{BernsteinK08dso}), when considering increasing intervals in tree $T_x$, such an endpoint must be the first center of priority $\ge p_v$ on the $x$-$v$ shortest path. These increasing intervals end at a priority $(p_v+1)$ vertex, and are excluded one at a time in line~\ref{alg:dsoquery:longcompute}. There are $O(\sqrt{n})$ such priority $(p_v+1)$ vertices, and thus $\tilde{O}(\sqrt{n})$ long interval computations per vertex $x$. Over all vertices $x \in V$, we broadcast $\tilde{O}(n^{3/2})$ words and perform $\tilde{O}(n^{3/2})$ exclude computations. For short intervals, in lines~\ref{alg:dsoquery:shorttrst}-\ref{alg:dsoquery:shorttren}, we only perform local computations at the end vertex $y$ of each shortest path.

    Using the algorithm for excluded shortest paths, Algorithm~\ref{alg:exclude}, we perform all excludes in lines~\ref{alg:dsoquery:mainexclude},\ref{alg:dsoquery:longcompute} in $\tilde{O}(n^{3/2})$ rounds (see Theorem~\ref{thm:dso:exclude}). The rest of the preprocessing consists of local computations, so the total round complexity for preprocessing is $\tilde{O}(n^{3/2})$ rounds.

    The query algorithm broadcasts $O(1)$ values and then broadcasts the correct answer, so it takes $O(D)$ rounds with $O(1)$ congestion per edge. 
    Thus, we can answer a batch of $k$ queries by pipelining all broadcasts in $O(k+D)$ rounds.

    \textbf{Oracle Size:} As indicated in line~\ref{alg:dsoquery:storage}, at vertex $y$, we store information from the modified APSP including $d(x,y), P(x,y), \CC{x}{y}{j}, \CC{y}{x}{j}$, which is $\tilde{O}(n)$ size since we have $n$ vertices $x$ and $1 \le j \le \log n$. Apart from this, we store $d(v,y,e)$ for edges $e$ within $\tilde{\Theta}(2^{p_v})$. As mentioned in the round complexity discussion, there are $\tilde{\Theta}(n/2^{p_v})$ vertices of priority $p_v$ \whp in $n$, and storing $\tilde{\Theta}(2^{p_v})$ for each such vertex gives size $\tilde{O}(n)$. Finally, we store $D(x,y,(c_x,c_y))$ values. Along any $x$-$y$ path, there are at most $2\log n$ intervals, and thus $y$ stores $O(\log n)$ distances per vertex $x$, leading to size $\tilde{O}(n)$. Thus, the oracle size at a given vertex $y$ is $\tilde{O}(n)$, and the total size over all vertices is $\tilde{O}(n^2)$.
\end{proof}

\subsection{DSO with Fast Preprocessing}
\label{sec:dso:algpre}

In this section, we present Algorithm~\ref{alg:dso:preproc} to construct a DSO that uses $\tilde{O}(n)$ rounds for preprocessing, stores $\tilde{O}(n)$ space per node, and answers a batch of $k$ queries in $\tilde{O}(k \cdot \sqrt{n}+D)$ rounds. The DSO stores a total of $\tilde{O}(n^2)$ words across all nodes, which is competitive with the sequential DSO algorithm of~\cite{BernsteinK09dso} that uses $O(n^2 \log n)$ space.

Our algorithm takes an approach that is different from the interval-based  framework of sequential results~\cite{DemetrescuTCR08dso,BernsteinK08dso,BernsteinK09dso}, and the distributed Algorithm~\ref{alg:dsoquery} in the previous section. Our method entirely avoids the use of excluded shortest paths computations.
Instead, our algorithm builds on a graph sampling technique used in a sequential sub-$n^3$ DSO construction in~\cite{WeimannY13}, which we leverage to compute hop-limited replacement paths via SSSP computations in sampled graphs. 

We further augment the graph sampling technique with a multi-level strategy that samples vertices as well as graphs.
We vary the graph and vertex sampling parameters
across
$\frac{1}{2}\log n$ levels to capture replacement paths of different lengths, which we compute using SSSP computations from sampled vertices in the constructed sampled graphs. 

We first describe the graph sampling method and then describe our distributed algorithm.

\subsubsection{Hop-limited Replacement Paths} \label{sec:dso:hoplimrp}
Given a hop parameter $h$, and a graph $G=(V,E)$ along with a set of sources $S \subseteq V$, we wish to compute $h$-hop replacement path distances $d_h(s,x,e)$ for $s \in S, x \in V, e \in E$. Here, $d_h(s,x,e)$ denotes the minimum weight of an $s$-$x$ path containing up to $h$ edges that does not contain $e$. 

\begin{lemma}
    Given a graph $G=(V,E)$ and $s,x \in V, e \in E$, let $2 \le h \le n$ be a hop-length parameter.
    Sample $\tilde{h}=20 h \log n$ graphs $G_1,\cdots, G_{\tilde{h}}$ as subgraphs of $G$ by including each edge of $G$ with probability $1-\frac{1}{h}$. Then,
    \begin{enumerate}[label=\Alph*.]
        \item There are $\Theta(\log n)$ graphs $G_i$ such that $e \not\in G_i$, \whp in $n$.
        \item Let $P_{sx}^e$ be a $h$-hop replacement path from $s$ to $x$ avoiding $e$. There is at least one graph $G_i$ such that $G_i$ contains all edges in $P_{sx}^e$ but does not contain $e$, \whp in $n$.
    \end{enumerate}
    \label{lem:dso:sampgraph}
\end{lemma}
\begin{proof}
The proof adapts arguments from~\cite{WeimannY13}.

    A. For a fixed edge $e$ and sampled graph $G_i$, $e \not\in G_i$ with probability $\frac{1}{h}$. Let $f_e$ be the number of graphs $G_i$ not containing $e$, then $E[f_e] = \frac{h'}{h} = 20\log n$. Using a Chernoff bound, $Pr[f_e \ge 2E[f_e] ] \le e^{-E[f_e]/3} \le n^{-6}$. Using a union bound over the $O(n^2)$ edges, $f_e \le 40\log n$ for all edges $e$ \whp in $n$.

    B. For a fixed edge $e$, and a path $P_{sx}^e$ of $h$ hops, the probability that a sampled graph $G_i$ does not contain $e$ but does contain all $h$ edges of $P_{sx}^e$ is $\frac{1}{h} \left( 1-\frac{1}{h}\right)^h \ge \frac{1}{4h}$ for $h \ge 2$. The probability that at least one of the $\tilde{h}=20 h \log n$ graphs $G_i$ satisfies this property is 
        $1-\left(1-\frac{1}{4h}\right)^{\tilde{h}} \ge 1-e^{-\frac{\tilde{h}}{4h}} \ge 1-n^{-5}$.
        Using a union bound over all $O(n^4)$ tuples $(e,s,x)$, we obtain our result \whp in $n$.
\end{proof}

\subsubsection{CONGEST DSO Algorithm with Fast Preprocessing}

We now describe Algorithm~\ref{alg:dso:preproc}.
To obtain fast preprocessing, we only deal with replacement paths of hop-length $\ge \sqrt{n}$ during preprocessing. Shorter replacement paths are computed at query time using $\sqrt{n}$ steps of Bellman-Ford with the appropriate edge deleted (line~\ref{alg:dso:preproc:bf}). 

To compute replacement paths of hop-length $[2^j,2^{j+1}]$, for $j\ge\frac{1}{2}\log n$, we use the result in Lemma~\ref{lem:dso:sampgraph} with $h=2^{j+1}$. After sampling vertices (line~\ref{alg:dso:preproc:samp}) and graphs (line~\ref{alg:dso:preproc:sampgraph}), we compute shortest paths to and from sampled vertices in line~\ref{alg:dso:preproc:sssp}. With this information, we can compute replacement path distances at query time (lines~\ref{alg:dso:preproc:edgesetbroad}-\ref{alg:dso:preproc:querysamp}) without needing to compute APSP during preprocessing. Although we could use the graph sampling method for $\le \sqrt{n}$ hop-length replacement paths and stay within the $\tilde{O}(n)$ preprocessing bound, the number of sampled vertices would exceed $\sqrt{n}$ and answering a query by broadcast (line~\ref{alg:dso:preproc:querybroadcast}) would take $\omega(\sqrt{n})$ rounds. We now prove that Algorithm~\ref{alg:dso:preproc} constructs the DSO in Theorem~\ref{thm:dso:algpre}, restated below.

\begin{algorithm}[!t]
    \caption{DSO with fast preprocessing: $\tilde{O}(n)$ preprocessing and $\tilde{O}(\sqrt{n}+D)$ query}
    \begin{algorithmic}[1]
        \Require Graph $G=(V,E)$.
        \Ensure Answer query for replacement path distance $d(x,y,e)$ for any $x,y \in V, e \in E$.
        \lineComment{\underline{Preprocessing Algorithm}}
        \For{$j=\log (\sqrt{n}), \cdots, \log n$} \label{alg:dso:preproc:prepst}
            \lineComment{Preprocessing for replacement paths of hop-length $[2^j,2^{j+1}]$, where $2^j \ge \sqrt{n}$}
            \State Each vertex $s \in V$ is sampled and added to $S_j$ with probability $\frac{c \log n}{2^j}$. \label{alg:dso:preproc:samp}
            \State Let $h=2^{j+1}$. Sample $\tilde{O}(h)$ graphs $G^j_i$ as in Lemma~\ref{lem:dso:sampgraph} by randomly sampling each edge with probability $1-(1/h)$. \label{alg:dso:preproc:sampgraph} \rightComment{Each edge locally samples itself}
            \State For each $e \in G$, $I_e^j \gets \{(j,i) \mid e \not\in G^j_i\}$  \label{alg:dso:preproc:edgesets} \rightComment{Computed locally at endpoints of $e$}
            \State Compute SSSP from each vertex in $S_j$, on each $G^j_i$ (and reversed graph). The distance in graph $G^j_i$ is denoted $d^j_i$. Distances $d^j_i(s,x), d^j_i(x,s)$ are computed at $x$ for each $s \in S_j,x \in V$. \label{alg:dso:preproc:sssp}
        \EndFor
        \lineComment{\underline{Query Algorithm}: Given query $d(x,y,e)$}
        \lineComment{Find best replacement path of hop-length $\le \sqrt{n}$}
        \State Perform a distributed Bellman-Ford computation in $G-\{e\}$ starting from $x$, up to $\sqrt{n}$ hops. This computes $d_{\sqrt{n}}(x,y,e)$ at $y$. \label{alg:dso:preproc:bf}
        \lineComment{Find best replacement path of hop-length $> \sqrt{n}$}
        \For{$j=\log (\sqrt{n}), \cdots, \log n$}
            \lineComment{Handle paths of hop-length $[2^j,2^{j+1}]$}
            \State Endpoint of $e$ broadcasts $I^j_e$. \label{alg:dso:preproc:edgesetbroad}
            \State $x$ broadcasts $\{d^j_i(x,s) \mid (j,i) \in I^j_e, s \in S_j\}$, i.e. its distance to sampled vertices in $S_j$ in graphs where $e$ is not present. \label{alg:dso:preproc:querybroadcast}
            \For{$s \in S_j$}   \rightComment{Local computation at $y$}
            \lineComment{$d_s$ is the best $x$-$y$ replacement path distance through $s$}
            \State Compute $d_s = \min\limits_{(j,i) \in I_e^j} d^j_i(x,s) + d^j_i(s,y)$  \label{alg:dso:preproc:querysamp}
            \EndFor 
        \EndFor
        \State Compute $d(x,y,e) = \min(d_{\sqrt{n}}(x,y,e), \min\limits_{s \in V} d_s$) \rightComment{Local computation at $y$}  \label{alg:dso:preproc:querymin}
    \end{algorithmic}
    \label{alg:dso:preproc}
\end{algorithm}

\thmdsoalgpre*
\begin{proof}
    \textbf{Correctness:} For query $d(x,y,e)$, if the $x$-$y$ replacement path has hop-length $\le \sqrt{n}$, the distance is correctly computed in line~\ref{alg:dso:preproc:bf} using the Bellman-Ford computation from $x$ up to $\sqrt{n}$ hops. We now address replacement paths of hop-length $[2^j,2^{j+1}]$, for $2^j \ge \sqrt{n}$.

    Fix one such replacement path $P$ from $x$ to $y$ avoiding edge $e$, with hop length $2^j \le h(P) < 2^{j+1}$ and weight $w(P)$. In the preprocessing algorithm, consider the $j$'th iteration of the loop (lines \ref{alg:dso:preproc:prepst}-\ref{alg:dso:preproc:sssp}). Since we sample vertices into $S_j$ with probability $\tilde{O}(1/2^j)$, \whp in $n$ there is at least one vertex $s$ from $S_j$ that is on ${P}$. By Lemma~\ref{lem:dso:sampgraph}.B at least one of the sampled subgraphs $G^j_i$ contains all edges of ${P}$ and does not contain $e$ \whp in $n$. 
    In this particular $G^j_i$, the distances $d^j_i(s,y)$ and $d^j_i(x,s)$ are computed in line~\ref{alg:dso:preproc:sssp} and we have $w({P}) = d^j_i(x,s)+d^j_i(s,y)$.
    
    When answering query $d(x,y,e)$, we use the set $I^j_e$ computed in line~\ref{alg:dso:preproc:edgesets} and broadcast in line~\ref{alg:dso:preproc:edgesetbroad}. In line~\ref{alg:dso:preproc:querysamp}, we only consider shortest path distances in graphs not containing $e$, so all distances $d_s$ correspond to valid $x$-$y$ paths not containing $e$. For the particular $G^j_i$ that contains ${P}$ and not $e$, we correctly compute $d(x,y,e) = d^j_i(x,s)+d^j_i(s,y)$ in line~\ref{alg:dso:preproc:querysamp}. 

    \textbf{Round Complexity:} We first analyze the preprocessing algorithm. Consider an iteration $j$: lines~\ref{alg:dso:preproc:samp}-\ref{alg:dso:preproc:edgesets} only involve local computation, and line~\ref{alg:dso:preproc:sssp} performs SSSP from $|S_j| = \tilde{O}(n/2^j)$ vertices on $\tilde{O}(2^j)$ graphs. Using the $\tilde{O}(n)$-round low congestion SSSP algorithm from~\cite{GhaffariT24sssp}, the total congestion over all SSSP computations is at most $\tilde{O}(|S_j| \cdot 2^j) = \tilde{O}(n)$. Thus, using random scheduling, line~\ref{alg:dso:preproc:sssp} takes $\tilde{O}(n)$ rounds.

    Now, to answer a query $d(x,y,e)$, computing Bellman-Ford up to $\sqrt{n}$ hops in line~\ref{alg:dso:preproc:bf} takes $O(\sqrt{n})$ rounds. The set $I^j_e$ with size $O(\log n)$ is broadcast in $\tilde{O}(D)$ rounds in line~\ref{alg:dso:preproc:edgesetbroad}. Vertex $x$ then broadcasts a set of size $|I^j_e|\cdot |S_j| = \tilde{O}(\frac{n}{2^j}) = \tilde{O}(\sqrt{n})$ (since $2^j \ge \sqrt{n})$ in line~\ref{alg:dso:preproc:querybroadcast}, which takes $\tilde{O}(\sqrt{n}+D)$ rounds. Lines~\ref{alg:dso:preproc:querysamp},\ref{alg:dso:preproc:querymin} only involve local computation at $y$, so the total rounds to answer the query is $\tilde{O}(\sqrt{n}+D)$. To answer a batch of $k$ queries, we can broadcast the $\tilde{O}(k\sqrt{n})$ required values in lines~\ref{alg:dso:preproc:edgesetbroad},\ref{alg:dso:preproc:querybroadcast} in $\tilde{O}(k\sqrt{n}+D)$ rounds, and perform $k$ Bellman-Ford computations sequentially in $k\sqrt{n}$ rounds -- totaling $\tilde{O}(k\sqrt{n}+D)$ rounds.

    \textbf{Space Usage:} In line~\ref{alg:dso:preproc:sssp}, each vertex $v$ stores distance $d_i(s,v)$ for each sampled vertex $s \in S_h$ and each sampled graph $G_i$. In a single iteration $j$, this is $\tilde{O}(h \cdot \frac{n}{h}) = \tilde{O}(n)$ distances. Over all the $O(\log n)$ iterations, each vertex stores a total of $\tilde{O}(n)$ distances.

\end{proof}

\section{Lower bounds}
\label{sec:lbs}
We now present CONGEST lower bounds for $k$-SSP and seb-DSO in directed unweighted graphs and undirected weighted graphs.

\subsection{\texorpdfstring{Lower Bound for $k$-SSP}{Lower Bound for k-SSP}}
\label{sec:kssp:lb}

We establish the CONGEST lower bound of $\tilde{\Omega}(\sqrt{nk}+D)$ for computing $k$-SSP in directed unweighted graphs or undirected weighted graphs stated in Theorem~\ref{thm:kssp:lb}. The directed unweighted result trivially applies to directed weighted graphs. Note that the undirected weighted bound is nearly tight for $(1+\epsilon)$-approximation, given the algorithm of~\cite{ElkinN19}. For undirected unweighted $k$-SSP, there are optimal $O(k+D)$ round algorithms~\cite{LenzenPP19,HoangPDGYPR19}. 

Our lower bound holds against any $\alpha$-approximation algorithm for directed $k$-SSP (for constant $\alpha>1$), and against $(2-\epsilon)$-approximation for undirected weighted $k$-SSP. In Definition~\ref{def:kssp} of the $k$-SSP problem, we mentioned that $d(x,y)$ must be computed at node $y$. Our lower bound applies to a more general model, where $d(x,y)$ may be computed at any node in $G$. 

We first state our lower bound for the directed case in Section~\ref{sec:kssp:lbdir}. Then, using a modified construction that uses weighted edges instead of directed edges, we obtain our lower bound for undirected weighted graphs in Section~\ref{sec:kssp:lbunw}. Our lower bound is for the more general problem of computing distances between $k$ given vertex pairs, which we first state for directed graphs in Theorem~\ref{thm:kssp:lb2}. The lower bound for $k$-SSP readily follows from this result. Our proof uses a reduction from the communication complexity problem of the set disjointness problem to computing distances between $k$ vertex pairs. 

We build on a construction of~\cite{SarmaHKKNPPW12}, which gives a $\tilde{\Omega}(\sqrt{n}+D)$ CONGEST lower bound for undirected weighted SSSP. While this result has been used to prove other $\sqrt{n}$-lower bounds, our method appears to give the first lower bound between $n$ and $\sqrt{n}$ for a shortest path problem in the CONGEST model.

\subsubsection{\texorpdfstring{Directed Unweighted $k$-SSP Lower Bound}{Directed Unweighted k-SSP Lower Bound}}
\label{sec:kssp:lbdir}

\begin{figure}[t]
    \centering
    \tikzstyle{vertex}=[circle, draw=black,minimum size=20pt]
    \tikzstyle{vertexsm}=[circle, draw=black,minimum size=5pt]
    \scalebox{0.6}{
    \begin{tikzpicture}
    
        \node[vertexsm] (u00) at (7,0) {};
        \node[vertexsm] (u0p2) at (3,-2) {};
        \node[vertexsm] (u0p1) at (1,-3.5) {};
        \node[vertexsm] (u1p1) at (5,-3.5) {};
        \node[vertexsm] (u1d) at (11,-3.5) {};

        \node (ul) at (4,-1.5) {\dots};
        \node (ur) at (9,-1.5) {\dots};
        \node (V) at (8,-8) {\dots};
        \node (V) at (8,-7) {\dots};
        \node (V) at (8,-9) {\dots};
        \node (V) at (8,-10) {\dots};

        \node[vertex] (up0) at (0,-5) {$u_0$};
        \node[vertex] (up1) at (2,-5) {$u_1$};
        \node[vertex] (up2) at (4,-5) {$u_2$};

        \node[vertex] (upd) at (12,-5) {$u_\ell$};

        \node[vertex] (v10) at (0,-7) {$v_{0}^1$};
        \node[vertex] (v11) at (2,-7) {$v_{1}^1$};
        \node[vertex] (v12) at (4,-7) {$v_{2}^1$};
        \node[vertex] (v1d1) at (10,-7) {};
        \node[vertex] (v1d) at (12,-7) {$v_{\ell}^1$};

        \node[vertex] (v20) at (0,-8) {$v_{0}^2$};
        \node[vertex] (v21) at (2,-8) {$v_{1}^2$};
        \node[vertex] (v22) at (4,-8) {$v_{2}^2$};
        \node[vertex] (v2d1) at (10,-8) {};
        \node[vertex] (v2d) at (12,-8) {$v_{\ell}^2$};

        \node[vertex] (vg0) at (0,-10) {$v_{0}^q$};
        \node[vertex] (vg1) at (2,-10) {$v_{1}^q$};
        \node[vertex] (vg2) at (4,-10) {$v_{2}^q$};
        \node[vertex] (vgd1) at (10,-10) {};
        \node[vertex] (vgd) at (12,-10) {$v_{\ell}^q$};


        \node[vertex] (s1) at (-2.5,-6.5) {$a_1$};
        \node[vertex] (s2) at (-2.5,-7.5) {$a_2$};
        \node[vertex] (s3) at (-2.5,-8.5) {$a_3$};
        \node[vertex] (s4) at (-2.5,-10.5) {$a_k$};

        \node[vertex] (b1) at (14.5,-6.5) {$b_1$};
        \node[vertex] (b2) at (14.5,-7.5) {$b_2$};
        \node[vertex] (b3) at (14.5,-8.5) {$b_3$};
        \node[vertex] (b4) at (14.5,-10.5) {$b_k$};

        \path[draw,thick,->] (u0p2) edge (u0p1);
        \path[draw,thick,->] (u0p2) edge (u1p1);
        \path[draw,thick,->] (u0p1) edge (up0);
        \path[draw,thick,->] (u0p1) edge (up1);
        \path[draw,thick,->] (u1p1) edge (up2);

        \path[draw,thick,<-] (up0) edge [bend left] (v20);
        \path[draw,thick,<-] (up0) edge [bend left] (vg0);

        \path[draw,thick,<-] (up1) edge [bend left] (v11);
        \path[draw,thick,<-] (up1) edge [bend left] (v21);
        \path[draw,thick,<-] (up1) edge [bend left] (vg1);
        \path[draw,thick,<-] (up2) edge [bend left] (v12);
        \path[draw,thick,<-] (up2) edge [bend left] (v22);
        \path[draw,thick,<-] (up2) edge [bend left] (vg2);

        \path[draw,thick,<-] (upd) edge [bend right] (v1d);
        \path[draw,thick,<-] (upd) edge [bend right] (v2d);

        \path[draw,thick,->] (v10) edge (v11);
        \path[draw,thick,->] (v11) edge (v12);
        \path[draw,thick,->] (v12) edge (6,-7);
        \path[draw,thick,->] (v1d1) edge (v1d);

        \path[draw,thick,->] (v20) edge (v21);
        \path[draw,thick,->] (v21) edge (v22);
        \path[draw,thick,->] (v22) edge (6,-8);
        \path[draw,thick,->] (v2d1) edge (v2d);

        \path[draw,thick,->] (vg0) edge (vg1);
        \path[draw,thick,->] (vg1) edge (vg2);
        \path[draw,thick,->] (vg2) edge (6,-10);
        \path[draw,thick,->] (vgd1) edge (vgd);

        \path[draw,thick,->] (u00) edge  (5.5,-1);
        \path[draw,thick,->] (u00) edge  (8.5,-1);
        \path[draw,thick,<-] (upd) edge  (u1d);
        \path[draw,thick,->] (u1d) edge  (10,-5);
        \path[draw,thick,<-] (u1d) edge  (10,-2);

        \path[draw,thick,->, dashed] (s1) edge (v10);
        \path[draw,thick,->, dashed] (s1) edge (v20);
        \path[draw,thick,->, dashed] (s1) edge (vg0);

        \path[draw,thick,->, dashed] (s2) edge (v10);
        \path[draw,thick,->, dashed] (s2) edge (v20);
        \path[draw,thick,->, dashed] (s2) edge (vg0);

        \path[draw,thick,->, dashed] (s4) edge (v10);
        \path[draw,thick,->, dashed] (s4) edge (v20);
        \path[draw,thick,->, dashed] (s4) edge (vg0);

        \path[draw,thick,->, dashed] (v1d) edge (b1);
        \path[draw,thick,->, dashed] (v2d) edge (b1);
        \path[draw,thick,->, dashed] (vgd) edge (b1);
        \path[draw,thick,->, dashed] (v1d) edge (b2);
        \path[draw,thick,->, dashed] (v2d) edge (b2);
        \path[draw,thick,->, dashed] (vgd) edge (b2);
        \path[draw,thick,->, dashed] (v1d) edge (b4);
        \path[draw,thick,->, dashed] (v2d) edge (b4);
        \path[draw,thick,->, dashed] (vgd) edge (b4);
    \end{tikzpicture}
    }
    \caption{Lower bound construction for directed unweighted $k$-SSP}
    \label{fig:kssp:lb}
\end{figure}

Given $n$ and $k$, $1 \le k \le n$, our goal is to obtain a lower bound for directed unweighted $k$-SSP on an $n$-node graph. Our reduction uses the following directed lower bound graph construction, based on inputs $S_a, S_b$ for set disjointness, with $\tilde{\Theta}(\sqrt{nk})$-bits each.

\textbf{Graph Construction:} We construct directed unweighted $n$-node graph $G=(V,E)$, pictured in Figure~\ref{fig:kssp:lb}, as a balanced binary tree with $\ell = \tilde{\Theta}(\sqrt{nk})$ leaves, numbered $u_0,u_1,\dots u_{\ell}$, along with a set of $q=\tilde{\Theta}(\sqrt{n/k})$ paths $v^i_0$-$v^i_{\ell}$ of length $\ell$. The paths are connected to the leaves of the tree as in the figure, with all edges oriented from path to leaf. 
Additionally, we have $2k$ vertices $a_i,b_i$ for $1 \le i \le k$. We perform a reduction from $(kq)$-bit set disjointness, and the dashed edges in our construction depend on the inputs $S_a,S_b$ of the set disjointness instance. We index the input such that $S_a[i,j]$ represents the $((i-1)\cdot q+j)$'th bit for $1\le i\le k, 1\le j \le q$ (also for $S_b$), similar to a directed RPaths lower bound construction in~\cite{ManoharanR24rp}. The edge $(a_i,v^j_0)$ is present only if $S_a[i,j]=1$; similarly  edge $(v^j_{\ell},b_i)$ is present only if $S_b[i,j]=1$. All other solid edges in Figure~\ref{fig:kssp:lb} are always present. The constant terms in $\ell, q$ are adjusted so that $G$ has $n$ vertices.

First, we prove the following connection between the set disjointness input and distances in $G$: 

\begin{lemma}
    In $G$, $d(a_i,b_i) =\ell+2$ if there exists a $j$, $1 \le j \le q$, such that $S_a[i][j]=S_b[i][j]=1$ and $d(a_i,b_i) = \infty$ otherwise. \label{lem:kssp:lbdist}
\end{lemma}
\begin{proof}
    Consider the shortest path from vertex $a_i$ to vertex $b_i$ for each $1 \le i \le k$. If both edges $(a_i,v^j_0)$ and $(v^j_{\ell},b_i)$ are present in $G$ for some $1 \le j \le q$, then we can use the $j$'th path to get $d(a_i,b_i) =\ell+2$. Any path from $a_i$ to $b_i$ must use the entirety of one of the $q$ paths of length $\ell$, since there is no way to move between paths using a directed edge. The edges to trees nodes are also directed away from reaching $b_i$, so no other paths are possible, proving the result.
\end{proof}

We now show how we can solve set disjointness using an algorithm for $k$ vertex pair distances: 

\begin{lemma}
    Given an algorithm $\mathcal{A}$ that computes $k$ distances $d(a_i,b_i)$ in $R \le (\ell/2)$ rounds, 
    we can solve $(kq)$-bit set disjointness using $Rp \cdot \Theta(\log n)$ bits, where $p$ is the height of the tree. \label{lem:kssp:lbsd}
\end{lemma}
\begin{proof}
    Our proof is similar to~(\cite{SarmaHKKNPPW12}, Theorem~3.1). Assume we are given an algorithm $\mathcal{A}$ that computes the $k$ distances $d(a_i,b_i)$ in $R \le (\ell/2)$ rounds. We show how to solve set disjointness for the inputs $S_a, S_b$. We omit some details here, while highlighting our modifications to the argument in~\cite{SarmaHKKNPPW12}.
    
    Consider an $i$-left set $L_i$ which consists of the first $(\ell-i)$ leaf vertices, path vertices $v_{j}^{r}$ for $1\le r\le q$ and $j \le \ell-i$, and all tree vertices that are the ancestor of at least one included leaf vertex. Similarly, an $i$-right set consists of the last $(\ell-i)$ leaf vertices, their ancestors in the tree, and path vertices $v_{j}^{r}$ for $1\le r\le q, j \ge i$. The idea is that for $i < \frac{\ell}{2}$, at round $i$, only a limited amount of information from $a_1,a_2,\dots a_k$ vertices reaches the vertices in $R_i$, and similarly only a limited information from $b_1,\dots, b_k$ vertices reaches vertices in $L_i$. More specifically, given the states (i.e. received messages and local computation) of vertices in $R_{i-1}$ at round $(i-1)$, we can compute the states of vertices in $R_i$ at round $i$ with an additional $\le p$ messages sent from vertices outside $R_{i-1}$ to vertices in $R_i$. A similar claim holds for $L_i$, but we will prove this only for $R_i$.

    Consider a neighbor $v$ of vertex $u$ in $R_i$. If $v$ is in $R_{i-1}$, then we know the state of $v$ at round $(i-1)$, and hence the message that it would have sent to $u$. Thus, we are only concerned with messages from neighbors $v$ that are not in $R_{i-1}$. As shown in~\cite{SarmaHKKNPPW12}, this only happens with tree ancestor vertices where a vertex is in $R_i$ but one of its children is not in $R_{i-1}$. There can be at most one such neighbor at each level of the tree (at the leftmost boundary of $R_{i}$), so there are at most $p$ of them. So, given the $\le p$ messages sent along these edges, the state of all vertices in $R_{i}$ at round $i$ can be determined.

    Now, we consider a simulation of the CONGEST algorithm by Alice and Bob. Initially, Alice has the bits $S_a$, and thus knows the edges incident to $a_j$ (for $1 \le j \le k$) and $v_0^r$ (for $1 \le r \le q$). Similarly, Bob has $S_b$ and knows the edges incident to $b_j$ and $v_{\ell}^r$. Initially, Alice knows the state of $L_1$ (which includes $v_0^r$'s but not $v_{\ell}^r$'s) and Bob knows the state of $R_1$. Observe that if $i < \frac{\ell}{2}$, $V \setminus R_{t-1} \subseteq L_{t-1}$.
    For each round of the CONGEST algorithm, Alice generates the $p = O(\log n)$ messages from $V \setminus R_{t-1}$ that Bob requires to compute $R_{t}$ from the state of $R_{t-1}$. Similarly, Bob generates the $p$ messages from $V\setminus L_{t-1} \subseteq R_{t-1}$ that Alice requires to compute $L_{t}$ from $L_{t-1}$. Thus, to simulate one round, Alice and Bob exchange $p \cdot \Theta(\log n)$ bits ($\Theta(\log n)$ term is due to the CONGEST bandwidth). To simulate $R$ rounds of the CONGEST algorithm, Alice and Bob communicate $\Theta(Rp\log n)$ bits.

    After this simulation concludes, either Alice or Bob knows each of the $k$ distances $d(a_i,b_i)$, since some vertex has computed this distance. If there is any $i$ such that $d(a_i,b_i) =\ell+2$, using Lemma~\ref{lem:kssp:lbdist}, we can determine that the sets $S_a, S_b$ are not disjoint. If one of Alice or Bob discovers this, they can communicate it to the other with only one bit. Thus, we have obtained a $\Theta(Rp\log n)$ bit communication protocol to determine set disjointness of $(kq)$-bit inputs $S_a, S_b$.
\end{proof}

Now, we are ready to prove our main result:
\begin{theorem}
    Given a directed unweighted graph $G=(V,E)$, and vertices $a_i,b_i \in V$ for $1 \le i \le k$, computing the $k$ distances $d(a_i, b_i)$ requires $\Omega\left(\frac{\sqrt{nk}}{\log n} + D\right)$ rounds, even in graphs of undirected diameter $\Theta(\log n)$. This lower bound applies to any $\alpha$-approximation algorithm for $k$-SSP (for constant $\alpha > 1$).
    \label{thm:kssp:lb2}
\end{theorem}
\begin{proof}
    Assume that we have an algorithm $\mathcal{A}$ that computes $k$ distances in a $n$-vertex graph in $R(n)$ rounds. We choose parameters $p = \Theta(\log n), q = \frac{n-k}{\ell} \ge \frac{n}{2\ell}$ (since $k<n/2$), $\ell = \frac{\sqrt{nk}}{\log n}$, and construct $G$ as described above. The number of vertices in $G$ is $\Theta(\ell \cdot q + k) = n$ and the undirected diameter of $G$ is $\le 2p+4 = \Theta(\log n)$ (using paths through the tree).

    The number of bits in the set disjointness instance is $kq \ge \frac{kn}{2\ell}$, and we apply Lemma~\ref{lem:kssp:lbsd} to obtain a set disjointness protocol which communicates $\Theta(Rp\log n)$ bits.  Using the randomized communication lower bound for set disjointness~\cite{Bar-YossefJKS04}, we obtain: If $R(n) \le (\ell/2)$, then $R(n) \cdot \log^2 n = \Omega(kq) \Rightarrow R(n) = \Omega\left(\frac{kn}{\ell\log^2 n}\right)$. Combining the two bounds, $R(n) = \Omega\left(\min(\ell, \frac{kn}{\ell\log^2 n})\right)$. We chose $\ell = \frac{\sqrt{nk}}{\log n}$ to balance these two expressions, giving $R(n) = \Omega\left(\frac{\sqrt{nk}}{\log n}\right)$.

    Since we are distinguishing between finite and infinite distances in Lemma~\ref{lem:kssp:lbdist}, the lower bound applies for any $\alpha$-approximation algorithm that computes these $k$ distances.
\end{proof}

The directed $k$-SSP result in Theorem~\ref{thm:kssp:lb}.\ref{thm:kssp:lb:dir} stated in Section~\ref{sec:results:lb} follows directly from Theorem~\ref{thm:kssp:lb2}, since computing $k$-SSP from sources $a_i$ also computes $d(a_i,b_i)$. We prove the results for undirected weighted graphs in the following section.

\subsubsection{\texorpdfstring{Undirected Weighted $k$-SSP Lower Bound}{Undirected Weighted k-SSP Lower Bound}}
\label{sec:kssp:lbunw}

We now adapt our directed graph construction to prove our $k$-SSP lower bound for undirected weighted graphs. The main idea is to use edge weights instead of directed edges to obtain the same reduction from set disjointness.

\begin{figure}[t]
    \centering
    \tikzstyle{vertex}=[circle, draw=black,minimum size=20pt]
    \tikzstyle{vertexsm}=[circle, draw=black,minimum size=5pt]
    \scalebox{0.6}{
    \begin{tikzpicture}
    
        \node[vertexsm] (u00) at (7,0) {};
        \node[vertexsm] (u0p2) at (3,-2) {};
        \node[vertexsm] (u0p1) at (1,-3.5) {};
        \node[vertexsm] (u1p1) at (5,-3.5) {};
        \node[vertexsm] (u1d) at (11,-3.5) {};

        \node (ul) at (4,-1.5) {\dots};
        \node (ur) at (9,-1.5) {\dots};
        \node (V) at (8,-8) {\dots};
        \node (V) at (8,-7) {\dots};
        \node (V) at (8,-9) {\dots};
        \node (V) at (8,-10) {\dots};

        \node[vertex] (up0) at (0,-5) {$u_0$};
        \node[vertex] (up1) at (2.5,-5) {$u_1$};
        \node[vertex] (up2) at (5,-5) {$u_2$};

        \node[vertex] (upd) at (12,-5) {$u_\ell$};

        \node[vertex] (v10) at (0,-7) {$v_{0}^1$};
        \node[vertex] (v11) at (2.5,-7) {$v_{1}^1$};
        \node[vertex] (v12) at (5,-7) {$v_{2}^1$};
        \node[vertex] (v1d1) at (10,-7) {};
        \node[vertex] (v1d) at (12,-7) {$v_{\ell}^1$};

        \node[vertex] (v20) at (0,-8) {$v_{0}^2$};
        \node[vertex] (v21) at (2.5,-8) {$v_{1}^2$};
        \node[vertex] (v22) at (5,-8) {$v_{2}^2$};
        \node[vertex] (v2d1) at (10,-8) {};
        \node[vertex] (v2d) at (12,-8) {$v_{\ell}^2$};

        \node[vertex] (vg0) at (0,-10) {$v_{0}^q$};
        \node[vertex] (vg1) at (2.5,-10) {$v_{1}^q$};
        \node[vertex] (vg2) at (5,-10) {$v_{2}^q$};
        \node[vertex] (vgd1) at (10,-10) {};
        \node[vertex] (vgd) at (12,-10) {$v_{\ell}^q$};


        \node[vertex] (s1) at (-2.5,-6.5) {$a_1$};
        \node[vertex] (s2) at (-2.5,-7.5) {$a_2$};
        \node[vertex] (s3) at (-2.5,-8.5) {$a_3$};
        \node[vertex] (s4) at (-2.5,-10.5) {$a_k$};

        \node[vertex] (b1) at (14.5,-6.5) {$b_1$};
        \node[vertex] (b2) at (14.5,-7.5) {$b_2$};
        \node[vertex] (b3) at (14.5,-8.5) {$b_3$};
        \node[vertex] (b4) at (14.5,-10.5) {$b_k$};

        \path[draw,thick,-] (u0p2) edge (u0p1);
        \path[draw,thick,-] (u0p2) edge (u1p1);
        \path[draw,thick,-] (u0p1) edge (up0);
        \path[draw,thick,-] (u0p1) edge (up1);
        \path[draw,thick,-] (u1p1) edge (up2);

        \path[draw,very thick,-] (up0) edge [bend left] node[midway, left] {$5nW$} (v10);
        \path[draw,very thick,-] (up0) edge [bend left] node[midway, right] {$5nW$} (v20);
        \path[draw,very thick,-] (up0) edge [bend left] node[midway, right] {$5nW$} (vg0);

        \path[draw,very thick,-] (up1) edge [bend left] node[midway, left] {$5nW$} (v11);
        \path[draw,very thick,-] (up1) edge [bend left] node[midway, right] {$5nW$} (v21);
        \path[draw,very thick,-] (up1) edge [bend left] node[midway, right] {$5nW$} (vg1);
        \path[draw,very thick,-] (up2) edge [bend left] node[midway, left] {$5nW$} (v12);
        \path[draw,very thick,-] (up2) edge [bend left] node[midway, right] {$5nW$} (v22);
        \path[draw,very thick,-] (up2) edge [bend left] node[midway, right] {$5nW$} (vg2);

        \path[draw,very thick,-] (upd) edge [bend right] node[midway, right] {$5nW$} (v1d);
        \path[draw,very thick,-] (upd) edge [bend right] node[midway, left] {$5nW$} (v2d);
        \path[draw,very thick,-,] (upd) edge [bend right] node[midway, left] {$5nW$} (vgd) (vgd);

        \path[draw,thick,-] (v10) edge (v11);
        \path[draw,thick,-] (v11) edge (v12);
        \path[draw,thick,-] (v12) edge (6,-7);
        \path[draw,thick,-] (v1d1) edge (v1d);

        \path[draw,thick,-] (v20) edge (v21);
        \path[draw,thick,-] (v21) edge (v22);
        \path[draw,thick,-] (v22) edge (6,-8);
        \path[draw,thick,-] (v2d1) edge (v2d);

        \path[draw,thick,-] (vg0) edge (vg1);
        \path[draw,thick,-] (vg1) edge (vg2);
        \path[draw,thick,-] (vg2) edge (6,-10);
        \path[draw,thick,-] (vgd1) edge (vgd);

        \path[draw,thick,-] (u00) edge  (5.5,-1);
        \path[draw,thick,-] (u00) edge  (8.5,-1);
        \path[draw,thick,-] (upd) edge  (u1d);
        \path[draw,thick,-] (u1d) edge  (10,-5);
        \path[draw,thick,-] (u1d) edge  (10,-2);

        \path[draw,thick,-, dashed] (s1) edge node[midway, above] {$W$} (v10);
        \path[draw,thick,-, dashed] (s1) edge node[midway, above] {$W$} (v20);
        \path[draw,thick,-, dashed] (s1) edge node[midway, above] {$W$} (vg0);

        \path[draw,thick,-, dashed] (s2) edge (v10);
        \path[draw,thick,-, dashed] (s2) edge (v20);
        \path[draw,thick,-, dashed] (s2) edge (vg0);

        \path[draw,thick,-, dashed] (s4) edge node[midway, above] {$ $} (v10);
        \path[draw,thick,-, dashed] (s4) edge node[midway, below] {$W$} (v20);
        \path[draw,thick,-, dashed] (s4) edge node[midway, below] {$W$} (vg0);

        \path[draw,thick,-, dashed] (v1d) edge node[midway, above] {$W$} (b1);
        \path[draw,thick,-, dashed] (v2d) edge node[midway, above] {$W$} (b1);
        \path[draw,thick,-, dashed] (vgd) edge node[midway, above] {$W$} (b1);
        \path[draw,thick,-, dashed] (v1d) edge (b2);
        \path[draw,thick,-, dashed] (v2d) edge (b2);
        \path[draw,thick,-, dashed] (vgd) edge (b2);
        \path[draw,thick,-, dashed] (v1d) edge node[midway, above] {$ $} (b4);
        \path[draw,thick,-, dashed] (v2d) edge node[midway, below] {$W$} (b4);
        \path[draw,thick,-, dashed] (vgd) edge node[midway, below] {$W$} (b4);
    \end{tikzpicture}
    }
    \caption{Lower bound construction for undirected weighted $k$-SSP. Thick edges (between tree leaves and path vertices) all have weight $5nW$, and dashed edges all have weight $W$ ($W = O(n)$ is a weight parameter).}
    \label{fig:kssp:lbunw}
\end{figure}

\textbf{Graph Construction:} We construct an undirected weighted graph $G=(V,E)$ shown in Figure~\ref{fig:kssp:lbunw} by modifying Figure~\ref{fig:kssp:lb}. The main idea is to use weighted edges instead of directed edges to obtain a reduction from set disjointness. The vertices and edges are the same in both figures, but all the edges are undirected and are weighted as follows, based on a weight parameter $W = O(n)$ to be chosen later: The edges between leaf vertices and path vertices $(u_i,v^j_i)$ have weight $5nW$, the edges between source/sink vertices and path vertices $(a_i,v^j_0)$ and $(v^j_{\ell},b_i)$ have weight $W$, and all other edges have unit weight. As in the directed construction, the edge $(a_i,v^j_0)$ is present only if $S_a[i,j]=1$, and $(v^j_{\ell},b_i)$ is present only if $S_b[i,j]=1$ (with weight $W$).

We prove the following lemma that is the basis of our reduction from set disjointness, similar to Lemma~\ref{lem:kssp:lbdist}.

\begin{lemma}
    In $G$, $d(a_i,b_i) =\ell+2W$ if there exists a $j$, $1 \le j \le q$, such that $S_a[i][j]=S_b[i][j]=1$ and $d(a_i,b_i) \ge \ell + 4W$ otherwise. \label{lem:kssp:lbdistunw}
\end{lemma}
\begin{proof}
    First, if $S_a[i][j]=S_b[i][j]=1$, then both edges $(a_i,v^j_0)$ and $(v^j_{\ell},b_i)$ are present, and we have a path from $a_i$ to $b_i$ using the $v^j_0$-$v^j_{\ell}$ path with weight $\ell+ 2W$. It can be readily seen that this is the shortest path (since using an edge from path to leaf would incur distance $5nW > \ell+2W$), and thus $d(a_i,b_i) =\ell+2W$.

    Now, assume that there is no $j$, $1 \le j \le q$, such that $S_a[i][j]=S_b[i][j]=1$. We need to prove $d(a_i,b_i) \ge \ell + 4W$. Consider a shortest path from $a_i$ to $b_i$ in this case. If this path uses a node to leaf edge, it has weight $> 5nW > \ell+4W$. If it doesn't, then this shortest path must use one of the paths $v_0^j$-$v_{\ell}^j$ in its entirety.

    Assume that the $a_i$-$b_i$ shortest path uses $v_0^j$-$v_{\ell}^j$ for some $j$. This part of the path incurs weight $\ell$, and must be connected to $a_i$ and $b_i$ via a $a_i$-$v_0^j$ path and a $b_i$-$v_{\ell}^j$ path. Since we have assumed $S_a[i][j]=0$ or $S_b[i][j]=0$, both edges  $(a_i,v^j_0)$ and $(v^j_{\ell},b_i)$ cannot exist in the graph. Assume $(a_i,v^j_0)$ doesn't exist. Then, the shortest path from $a_i$ to $v^j_0$ must use an edge from $a_i$ to a different $v^z_0$ ($z \ne j$), and then two or more weight $W$ edges to reach $v^j_0$. Thus, the $a_i$-$v^j_0$ path has weight at least $3W$. Combining this with the $b_i$-$v_{\ell}^j$ path of weight at least $W$, we get a total weight $\ell + 4W$ for the shortest path. Thus, $d(a_i,b_i) \ge \ell + 4W$.
\end{proof}

\begin{theorem}
    Given a undirected weighted graph $G=(V,E)$, and vertices $a_i,b_i \in V$ for $1 \le i \le k$, computing the $k$ distances $d(a_i, b_i)$ requires $\Omega\left(\frac{\sqrt{nk}}{\log n} + D\right)$ rounds, even in graphs of undirected diameter $\Theta(\log n)$. This lower bound applies for any $(2-\epsilon)$-approximation algorithm for $k$-SSP.
    \label{thm:kssp:lb2unw}
\end{theorem}
\begin{proof}
    We prove the following analogue of the directed Lemma~\ref{lem:kssp:lbsd} for our undirected weighted construction using Lemma~\ref{lem:kssp:lbdistunw}:

    {\textbf{Claim: }} Given an algorithm $\mathcal{A}$ for undirected weighted $k$-SSP that computes $(2-\epsilon)$ approximation of $k$ distances $d(a_i,b_i)$ in $R \le (\ell/2)$ rounds, 
    we can solve $(kq)$-bit set disjointness using $Rp \cdot \Theta(\log n)$ bits, where $p$ is the height of the tree.

    {\textbf{Proof of Claim: }}
    If the set disjointness input is not disjoint, we have $d(a_i,b_i) =\ell+2W$ for some $i, 1 \le i \le k$. If the set disjointness input is disjoint, we have $d(a_i,b_i) \ge \ell+4W$. Even if $\mathcal{A}$ is a $(2-\epsilon)$-approximation algorithm, it can distinguish between these two cases: In the first case, the output is $\le (2-\epsilon)(\ell+2W)$, and in the second case the output is $ \ge \ell+4W$. We choose $W = \ell(1-\epsilon)/(2\epsilon) = O(n/\epsilon)$ to ensure that $(2-\epsilon)(\ell+2W) < \ell+4W$, which makes the two cases distinguishable.

    Returning to the main proof, the above claim shows that we can solve set disjointness using the output of the $k$ shortest path queries in the constructed undirected weighted graph. Using a proof similar to Theorem~\ref{thm:kssp:lb2} gives us a $\Omega\left(\frac{\sqrt{nk}}{\log n} + D\right)$ lower bound for any algorithm $\mathcal{A}$ that computes $(2-\epsilon)$-approximate shortest paths between $k$ pairs.
\end{proof}

The undirected weighted $k$-SSP result in Theorem~\ref{thm:kssp:lb}.\ref{thm:kssp:lb:undir} stated in Section~\ref{sec:results:lb} follows from Theorem~\ref{thm:kssp:lb2unw}. For exact undirected unweighted $k$-SSP, optimal algorithms taking $O(k+D)$ rounds were already known~\cite{LenzenPP19,HoangPDGYPR19}.

\subsection{Lower Bound for DSO}
\label{sec:dso:lb}

We now present a CONGEST lower bound of $\tilde{\Omega}(\sqrt{nk}+D)$ for computing answers to $k$ seb-DSO queries in directed unweighted graphs (which trivially applies to directed weighted graphs), stated in Theorem~\ref{thm:dso:lb}. As with our $k$-SSP lower bound (Theorem~\ref{thm:kssp:lb} in Section~\ref{sec:results:lb}), this result also extends to undirected weighted graphs, as we detail in Section~\ref{sec:dso:lbunw}.
Similar to the $k$-SSP problem, these lower bounds apply to algorithms where $d(x,y)$ may be computed at any node in $G$. 

\begin{figure}[t]
    \centering
    \includegraphics[scale=0.6]{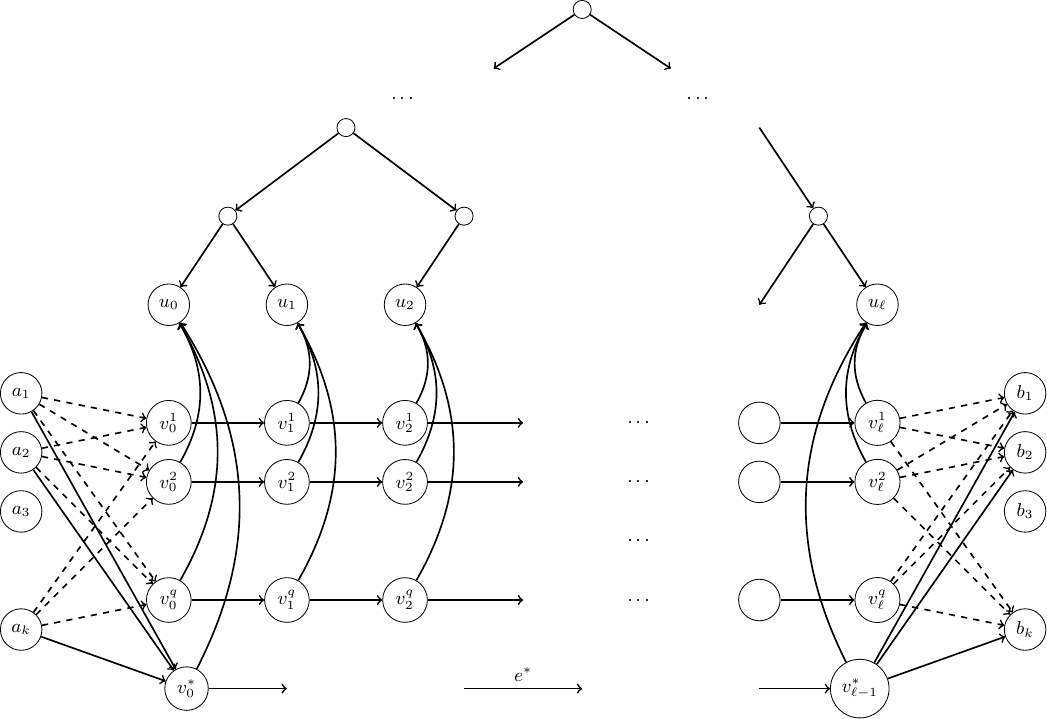}
    \caption{Lower bound for answering $k$ DSO queries in directed unweighted graphs}
    \label{fig:dso:lb}
\end{figure}

We build on the lower bound construction for directed $k$-SSP from Section~\ref{sec:kssp:lb}. We cannot directly connect the $a_i,b_i$ vertices with an added edge, since that would decrease the undirected diameter of the new graph and introduce additional pathways of communication, breaking the lower bound proof. So, we add an additional long path that connects $a_i, b_i$ vertices, as depicted in Figure~\ref{fig:dso:lb}. When an edge $e^*$ on this path is removed, the new shortest path corresponds to $d(a_i,b_i)$ in the original $k$-SSP lower bound construction. Thus, answering $k$ DSO queries $d(a_i,b_i,e^*)$ corresponds to computing $k$ shortest path distances $d(a_i,b_i)$ in the unmodified graph, leading to the following proof of Theorem~\ref{thm:dso:lb} (stated in Section~\ref{sec:results:lb}).

\thmdsolb*
\begin{proof}
    We use the construction in Figure~\ref{fig:dso:lb}, which modifies the construction in Theorem~\ref{thm:kssp:lb2} by adding an additional path $v^*_0$-$v^*_{\ell-1}$ of length $\ell-1$. This path is connected to vertices $a_i,b_i$ and the leaves $u_i$ as shown in the figure. 
    
    Note that $d(a_i,b_i)=\ell+1$ due to the $v^*_0$-$v^*_{\ell-1}$ path. However, when an edge $e^*$ on the $v^*_0$-$v^*_{\ell-1}$ path is removed, a replacement path from $a_i$ to $b_i$ with weight $\ell+2$ exists through $v^j_0$-$v^j_{\ell}$ if $S_a[i,j]=S_b[i,j]=1$ for some $j$, i.e., $d(a_i,b_i,e^*)=\ell+2$. 
Otherwise, if no such $j$ exists (i.e., for all $j$, $S_a[i,j]=0$ or $S_b[i,j]=0$), then for every path $v^j$ at least one of the edges $(a_i,v^j_0)$ or $(v^j_{\ell},b_i)$ is missing. In this case, no $a_i$-$b_i$ replacement path exists and $d(a_i,b_i,e^*) = \infty$. We thus prove the following claim, similar to Lemma~\ref{lem:kssp:lbdist}.

    \textbf{Claim 1:} Let $e^*$ be an edge on the $v^*_0$-$v^*_{\ell-1}$ path. $d(a_i,b_i,e^*) =\ell+2$ if $S_a[i][j]=S_b[i][j]=1$ and $d(a_i,b_i,e^*) = \infty$ otherwise. 

    Following a proof similar to Lemma~\ref{lem:kssp:lbsd}, building on the analysis of(\cite{SarmaHKKNPPW12}, Theorem~3.1), we get the following result.

    \textbf{Claim 2:} If there is an algorithm $\mathcal{A}$ that computes answer to $k$ queries $d(a_i,b_i,e^*)$ (all for the same failed edge $e^*$) in $R \le (\ell/2)$ rounds, then we can solve $(kq)$-bit set disjointness using $Rp \cdot \Theta(\log n)$ bits, where $p$ is the height of tree. 

    Applying the randomized communication lower bound for set disjointness, with parameters $p = \Theta(\log n), q = \frac{n-k}{\ell} \ge \frac{n}{2\ell}$ (since $k<n/2$), $\ell = \frac{\sqrt{nk}}{\log n}$, we get a lower bound of $\Omega\left(\frac{\sqrt{nk}}{\log n}\right)$. 
    
    For more details on Claim 2 and how we choose our parameters, see the proof of Theorem~\ref{thm:kssp:lb2}.

    Similar to Theorem~\ref{thm:kssp:lb2} for $k$-SSP, since we distinguish between finite and infinite distances in Claim 1, this lower bound holds for any approximate algorithm for DSO. Since Claim 1 is for computing replacement path distances for the same failed edge $e^*$, the lower bound applies to seb-DSO (and hence trivially for general DSO).
\end{proof}

\subsubsection{Undirected Weighted DSO Lower Bound.}
\label{sec:dso:lbunw}
The DSO lower bound proof in Theorem~\ref{thm:dso:lb} can be adapted to work for exact DSO in undirected weighted graphs as well. This also applies to $(2-\epsilon)$-approximation of undirected weighted DSO (unlike the directed result, where the lower bound holds for any constant approximation). The main idea is to replace the directed edges in Figure~\ref{fig:dso:lb} with weighted undirected edges, such that we achieve a separation in DSO distances similar to Claim 1 in the above proof of Theorem~\ref{thm:dso:lb}. Since this construction exactly follows the modifications we made for the $k$-SSP lower bound in Section~\ref{sec:kssp:lbunw}, we do not repeat it here for DSO. 

\subsubsection{Further DSO Lower Bounds.}
\paragraph{DSO without preprocessing}
The lower bound in Theorem~\ref{thm:dso:lb} with $P(n)=0$ shows that any algorithm that answers $k$ seb-DSO queries with no preprocessing requires $\tilde{\Omega}(\sqrt{nk}+D)$ rounds. Along with the upper bound result in Observation~\ref{thm:dso:seb} (in Section~\ref{sec:dso:sebintro}), this shows that both the CONGEST upper and lower bounds for seb-DSO without preprocessing match $k$-SSP. For $k \ge n^{1/3}$ in directed unweighted graphs, we have optimal bounds, achieving $\tilde{\Theta}(\sqrt{nk}+D)$ rounds. A similar tight relation between seb-DSO and $k$-SSP bounds holds for undirected graphs, in both the unweighted~\cite{LenzenPP19,HoangPDGYPR19} and the $(1+\epsilon)$-approximate weighted case~(\cite{ElkinN19},Theorem~\ref{thm:kssp:lb2unw} in Section~\ref{sec:kssp:lbunw}).

\begin{corollary}
    \label{lem:seblb}
    Any DSO algorithm that answers $k$ seb-DSO queries in a graph $G=(V,E)$ with no preprocessing requires $\tilde{\Omega}(\sqrt{nk}+D)$ rounds even if $G$ has undirected diameter $\Theta(\log n)$.
\end{corollary}

\paragraph{DSO with preprocessing}
We assume query costs similar to our DSO algorithms from Section~\ref{sec:dso:algs}, to obtain the lower bounds in Theorem \ref{thm:dso:prelb2} (restated below). The bounds follow from Theorem~\ref{thm:dso:lb} by choosing the appropriate $k$ specified below. In both cases, there is a $\sqrt{n}$ gap in preprocessing rounds for our upper and lower bounds.

\thmdsoprelbb*

The following Theorem~\ref{thm:dso:prelb} gives lower bounds for any DSO algorithm that takes $Q(n)$ rounds per query independent of $k$. It applies for $Q(n) = o\left( \frac{\sqrt{n}}{\log n}\right)$, and covers the salient range of query response complexities since the current best CONGEST lower bound for SSSP is $\Omega\left(\frac{\sqrt{n}}{\log n}+D\right)$ and SSSP can be used to answer queries with no preprocessing. Note that we could have derived Theorem~\ref{thm:dso:prelb2} as a specific case of Theorem~\ref{thm:dso:prelb}, by choosing $Q(n)$ to be $O(1)$ or $ {O}(\sqrt{n}/\log^2 n)$.

\begin{restatable}{theorem}{thmdsoprelb}
Consider a DSO algorithm for directed unweighted graph $G=(V,E)$ on $n$ nodes with diameter $D=\Omega(\log n)$ that performs $P(n)$ rounds of preprocessing. Assume the algorithm can answer a batch of $k$ queries in $O(kQ(n)+D)$ rounds for some fixed function $Q(n) = o\left( \frac{\sqrt{n}}{\log n}\right)$, and for a particular choice of $k= \frac{n}{Q(n)^2\log^3 n}$. Then, $P(n) = \omega(\frac{n}{Q(n)\log^3 n})$.
    \label{thm:dso:prelb}
\end{restatable}
\begin{proof}
    Applying Theorem~\ref{thm:dso:lb}, for any $k$, the algorithm takes $\Omega(\frac{\sqrt{nk}}{\log n})$ rounds on some graph $G$ with diameter $D=O(\log n)$. By our assumption, the algorithm takes $\le P(n) + O(kQ(n)+D)$ rounds for preprocessing combined with answering $k$ queries. So, $P(n) + k \cdot Q(n) = \Omega\left(\frac{\sqrt{nk}}{\log n}\right)$. 
    
    Thus, $P(n) = \Omega\left(\frac{\sqrt{nk}}{\log n} - kQ(n) - D\right)$. We choose $k$ such that $\frac{\sqrt{nk}}{\log n} = \omega(kQ(n))$, e.g. $k = \frac{1}{\log n} \cdot \frac{n}{Q(n)^2 \log^2n}$. Note that $k \ge 1$ due to our condition on $Q$.
    With this choice, $P(n) = \Omega(\frac{\sqrt{nk}}{\log n}) = \Omega(\frac{n}{Q(n)\log^2 n \sqrt{\log n}}) = \omega(\frac{n}{Q(n)\log^3 n})$.
\end{proof}

Theorems~\ref{thm:dso:prelb}, \ref{thm:dso:prelb2} address DSO algorithms with fixed cost per query. We also obtain lower bounds for two other query cost models, that capture other DSO query costs.

We consider a `batched' query model, where a DSO algorithm is designed for a particular batch size $\beta \ge 1$. Consider a DSO algorithm that computes the answers to $k$ queries in $Q'(n,k)$ rounds if $k \le \beta$, and if $k > \beta$, it takes $O(\frac{k}{\beta} Q'(n,\beta))$ rounds: This corresponds to processing queries $\beta$ at a time if $k > \beta$.
We obtain lower bounds for this model directly from Theorem~\ref{thm:dso:prelb} by setting $Q(n) = \frac{Q'(n,\beta)}{\beta}$.

Another query cost model we consider is of the form $Q(n,k) = O(\sqrt{n} k ^\alpha)$ for some constant $0 \le \alpha < \frac{1}{2}$. This is inspired by the fact that the seb-DSO algorithm with {\it no preprocessing} in Section~\ref{sec:dso:nopre} achieves $\alpha=1/2$, taking $\tilde{O}(\sqrt{nk})$ rounds when $k \ge n^{1/3}$. 
We show that achieving $\alpha < 1/2$ for general batched DSO for all $k$ requires $\Omega(n)$ rounds of preprocessing, using our general DSO lower bound from Theorem~\ref{thm:dso:lb}.

\begin{corollary} 
    If a DSO algorithm can answer $k$ queries in $O(\sqrt{n} k ^\alpha + D)$ rounds, for any $1 \le k \le n$, and for some fixed  constant $\alpha$, where $0 \le \alpha < \frac{1}{2}$, then the algorithm must use $\Omega(\frac{n}{\log n})$ preprocessing rounds.
\end{corollary}
\begin{proof}
    If $P(n)$ is the preprocessing cost, apply Theorem~\ref{thm:dso:lb} to get $P(n) + \sqrt{n} k^\alpha = \Omega(\frac{\sqrt{nk}}{\log n})$, which gives us a lower bound of $P(n) = \Omega(\frac{\sqrt{nk}}{\log n})$. Choosing $k=n$ gives a lower bound of $P(n) = \Omega(\frac{n}{\log n})$.
\end{proof}
\section{All Pairs Second Simple Shortest Paths (2-APSiSP)}
\label{sec:apsisp}
In this section, we prove our nearly optimal $\tilde{\Theta}(n)$ upper and lower bounds for 2-APSiSP in CONGEST.

\subsection{2-APSiSP Algorithm}
\label{sec:apsisp:ub}

We present Algorithm~\ref{alg:capsisp} for computing 2-APSiSP in directed weighted graphs in $\tilde{O}(n)$ rounds, 
proving Theorem~\ref{thm:apsisp:ub}. Our algorithm builds on a sequential 2-APSiSP algorithm from~\cite{AgarwalR16sisp}.
The results in~\cite{AgarwalR16sisp} are for the more general problem of computing $k$ simple shortest paths between all pairs of vertices ($k$-APSiSP). The $k$-APSiSP algorithm in~\cite{AgarwalR16sisp} adapts a path extension technique used in~\cite{DemetrescuI04} for fully dynamic APSP, and it runs in $\tilde{O}(mn)$ time for $k=2$. The running time increases with $k$, and is no longer $\tilde{O}(mn)$ for $k> 2$. We focus on $k=2$ here.

The algorithm is based on the following characterization of 2-SiSP from~\cite{AgarwalR16sisp}.

\begin{lemma}
    \label{lem:apsisp}~\cite{AgarwalR16sisp}
    Given $x,y \in V$, let $a$ be the vertex after $x$ on the $x$-$y$ shortest path. For $a \ne y$,
    2-SiSP distance $d_2(x,y) = \min\left( d(x,y,(x,a)), w(x,a) + d_2(a,y) \right)$
\end{lemma}

Lemma~\ref{lem:apsisp} shows that we can classify 2-SiSP paths into two types: a Type-A path, where the 2-SiSP distance $d_2(x,y)$ is the  distance $d(x,y,(x,a))$ that avoids the first edge $(x,a)$ on $x$-$y$ shortest path, and a Type-B path that uses edge $(x,a)$ but differs from the $x$-$y$ shortest path on a different edge. The 2-APSiSP algorithm computes these two types using different methods.

\begin{algorithm}[!ht]
    \caption{\congest 2-APSiSP Algorithm}
    \begin{algorithmic}[1]
        \Require Directed weighted graph $G=(V,E)$.
        \Ensure Compute 2-SiSP distance $d_2(x,y)$ at $y$ for each $x,y \in V$.
        \State Compute APSP in $G$ and $G^r$, remembering parent nodes. \label{alg:capsisp:apsp}
        \State  For each vertex $y \in V$, initialize $H_y = \emptyset$. $H_y$ is a priority queue for Type-B paths \textit{ending} at $y$. $H_y$ will contain keys $x \in V$ with priority $D$ representing $x$-$y$ 2-SiSP of weight $D$. \rightComment{$H_y$ is used in line~\ref{alg:capsisp:type2f}, lines~\ref{alg:capsisp:secst}-\ref{alg:capsisp:secen}.} 
        \ForAll{vertex $x \in V$} \label{alg:capsisp:initst}
            \State Find all neighbors $a$ such that edge $(x,a)$ is a shortest path. 
            \State Perform exclude computation with source $x$, and edges $(x,a)$ as the excluded set of independent paths using Algorithm~\ref{alg:exclude} (from Section~\ref{sec:exclude}): Distance $d(x,y,(x,a))$ is computed at $y$. \label{alg:capsisp:exclude} \rightComment{Used in lines~\ref{alg:capsisp:excludeset},\ref{alg:capsisp:type2f}}
            \State  Perform a downcast on $T_x$ (out-shortest path tree for source $x$), propagating $(x,a)$ and $w(x,a)$ to each vertex $y$ in the subtree rooted at $a$. \rightComment{Propagate first-edge info, used in lines~\ref{alg:capsisp:extendcreate},\ref{alg:capsisp:type2f}}  \label{alg:capsisp:downcast}
            \For{vertex $a$} 
            \For{vertex $y \in V$ in subtree of $T_x$ rooted at $a$} 
                \lineComment{Lines~\ref{alg:capsisp:excludeset}-\ref{alg:capsisp:initen}: Local computation at $y$. Sets Type-A distance in line~\ref{alg:capsisp:excludeset} and initializes Type-B distance computation.}
                \State Let $d_2^*(x,y) \gets d(x,y,(x,a))$ be the current estimate of 2-SiSP distance. \rightComment{Set Type-A distance, computed in line~\ref{alg:capsisp:exclude}} \label{alg:capsisp:excludeset}
                \State Add $(x,a)$ to $Extensions(a,y)$. \rightComment{Used in line~\ref{alg:capsisp:extend} to extend Type-B paths } \label{alg:capsisp:extendcreate}
                \lineComment{Add Type-B candidate for $x$-$y$ 2-SiSP to $H_y$.}
                \State If $(a,b)$ is the first edge of $a$-$y$ shortest path and $b \ne y$, add key $x$ to $H_y$ with weight $w(x,a) + d(a,y,(a,b))$. \rightComment{Note that $(a,b)$ has been propagated from the downcast with source $a$ from line~\ref{alg:capsisp:downcast}, and $d(a,y,(a,b))$ is computed at $y$ in line~\ref{alg:capsisp:exclude}.} \label{alg:capsisp:initen} \label{alg:capsisp:type2f}
            \EndFor
            \EndFor
        \EndFor 
        \lineComment{Lines \ref{alg:capsisp:secst}-\ref{alg:capsisp:secen}: Compute shortest Type-B 2-SiSP distance}
        \For{vertex $y \in V$} \label{alg:capsisp:secst}
        \lineComment{Local computation at $y$ to compute $d_2(x,y)$ for all $x \in V$.}
            \While{$H_y \ne \emptyset$}
            \State Let $(x,D) \gets \textsc{Extract-min}(H_y)$ \label{alg:capsisp:extract}
            \If{$D \ge d_2^*(x,y)$}
                \State Set $d_2(x,y) \gets d_2^*(x,y)$ \rightComment{Type-A path is shortest.} \label{alg:capsisp:found1}
            \Else
                \State Set $d_2(x,y) \gets D$ \rightComment{Type-B path is shortest.} \label{alg:capsisp:found}
                \For{$(x',x) \in Extensions(x,y)$}
                    \State Add key $x'$ with weight $w(x',x)+d_2(x,y)$ to $H_y$. \rightComment{Update Type-B candidate path distances for $x'$-$y$ 2-SiSP.} \label{alg:capsisp:extend} \label{alg:capsisp:secen}
                \EndFor
            \EndIf
            \EndWhile
        \EndFor
    \end{algorithmic}
    \label{alg:capsisp}
\end{algorithm}

We present our distributed 2-APSiSP method in Algorithm~\ref{alg:capsisp}. Our algorithm is a CONGEST implementation of a streamlined version of the $k$-APSiSP algorithm in~\cite{AgarwalR16sisp}, simplified for $k=2$. 
We make a few modifications to enable a distributed implementation: 
We perform a downcast through each shortest path tree to facilitate availability of information at the right nodes.
We break down the global priority queue used for Type-B distances into a separate priority queue for each end vertex in a shortest path, which allows us to independently compute distances locally.  These concerns are not present in the sequential setting, where information is available globally. 
As in our other algorithms, we only output 2-SiSP distances and not the entire path.

We now briefly describe the steps in Algorithm~\ref{alg:capsisp} and then address correctness and round bound in Theorem~\ref{thm:apsisp:ub}. We first perform APSP in line~\ref{alg:capsisp:apsp} to compute shortest path trees. We then initialize the queues $H_y$, which will contain candidate Type-B paths ending at vertex $y$. Lines~\ref{alg:capsisp:initst}-\ref{alg:capsisp:initen} are performed at each vertex simultaneously.

After identifying the first edges $(x,a)$ in the shortest path tree $T_x$ rooted at $x$, we compute Type-A path distances in lines~\ref{alg:capsisp:exclude} and~\ref{alg:capsisp:excludeset}. Line~\ref{alg:capsisp:exclude} involves $n$ exclude computations --- one from each vertex $x \in V$. Using Theorem~\ref{thm:dso:exclude} from Section~\ref{sec:exclude}, this takes $\tilde{O}(n)$ rounds.
We make use of downcasts in line~\ref{alg:capsisp:downcast} to propagate the first-edge information to vertices in $T_x$. This allows each node $y$ in the subtree of $T_x$ rooted at $a$ to learn the relevant first edge $(x,a)$ on the $x$-$y$ path, which is needed for Type-B computations.

For Type-B paths, we use a priority queue $H_y$ for candidate paths ending at vertex $y \in V$; these replace the global priority queue used in~\cite{AgarwalR16sisp}. This computation is entirely local at each $y$, using information received during the exclude (line~\ref{alg:capsisp:exclude}) and downcast (line~\ref{alg:capsisp:downcast}) phases.
The first path added to $H_y$ is in line~\ref{alg:capsisp:type2f}, which avoids the second edge on the shortest path (if one exists). Other candidates are added in line~\ref{alg:capsisp:extend}: Whenever a improved Type-B distance is found in line~\ref{alg:capsisp:found}, it is used to propagate new Type-B candidate paths for other sources $x'$ (identified in line~\ref{alg:capsisp:extendcreate}). Otherwise, if the Type-A distance is shortest and $d_2$ is set in line~\ref{alg:capsisp:found1}, the correct Type-B path has already been added in line~\ref{alg:capsisp:type2f}.

\thmapsispub*

\begin{proof}
    \textbf{Correctness:} Correctness follows from Lemma~\ref{lem:apsisp} and arguments in Lemma 2.2 in \cite{AgarwalR16sisp} that builds on the path extension method in~\cite{DemetrescuI04}: Type-A path distances are correctly computed in lines~\ref{alg:capsisp:exclude},\ref{alg:capsisp:excludeset}. The minimum distance found in line~\ref{alg:capsisp:found} is the correct shortest Type-B distance, since all smaller 2-SiSP distances of paths ending in $y$ have already been found. Thus, the shortest Type-B distance had been added to $H_y$ in line~\ref{alg:capsisp:extend} or line~\ref{alg:capsisp:type2f} before it is extracted in line~\ref{alg:capsisp:extract}. 

    \textbf{Round Complexity:} APSP in line~\ref{alg:capsisp:apsp} takes $\tilde{O}(n)$ rounds using the algorithm of~\cite{BernsteinN19apsp}. We perform exclude computations from $n$ sources in line~\ref{alg:capsisp:exclude} which takes $\tilde{O}(n)$ rounds using Theorem~\ref{thm:dso:exclude}. The downcast in line~\ref{alg:capsisp:downcast} takes $O(n)$ rounds and $O(1)$ congestion per edge, so the computation for all vertices can be randomly scheduled in $\tilde{O}(n)$ rounds. The rest of the computation is performed locally at each $y$. Thus, the total round complexity is $\tilde{O}(n)$ rounds.
\end{proof}

\subsection{2-APSiSP Lower Bound}
\label{sec:apsisp:lb}
We prove a lower bound of $\tilde{\Omega}(n)$ for 2-APSiSP, proving Theorem~\ref{thm:apsisp:lb}.
The $k$-SSP lower bound result in Theorem~\ref{thm:kssp:lb}, with $k=\Theta(n)$, can be modified to give a $\tilde{\Omega}(n)$ lower bound for 2-APSiSP. However, this result would only apply to directed graphs with diameter $\Omega(\log n)$. The lower bound presented below holds for undirected unweighted graphs with diameter as low as $\Theta(1)$. In fact, it holds even when APSP distances in the input graph are known, i.e., vertex $y \in V$ knows distances $d(x,y)$ for all $x \in V$.

\thmapsisplb*
\begin{proof}
    \begin{figure}[t]
        \centering
        \tikzstyle{vertex}=[circle, draw=black,minimum size=16pt]
        \includegraphics[scale=0.6]{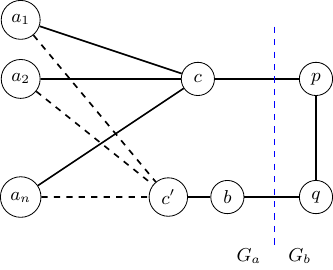}
        \caption{2-APSiSP lower bound.}
        \label{fig:apsisplb2}
    \end{figure}
    Consider the undirected unweighted graph construction $G=(V,E)$ in Figure~\ref{fig:apsisplb2}, based on a $n$-bit set disjointness instance $S_a, S_b$. Solid edges are always present, but for each $1 \le i \le n$, edge $(a_i,c')$ is present only if $S_a[i]=1$. $G$ has $\Theta(n)$ vertices and constant undirected diameter. Alice controls vertices $a_i,c,c',b$ and Bob controls vertices $p,q$. The cut separating their partitions consists of the two edges $(c,p)$ and $(b,q)$.
    Note that Alice and Bob can construct their sections of the graph with only the bits known to them, and the cut edges are fixed.

    First, we show that set disjointness can be determined by computing 2-APSiSP in $G$. Note that $d(a_i, q)=3$, either using the path through vertex $c$, or using the path through $c'$ when edge $(a_i,c')$ is present is $G$. If edge $(a_i,c')$ is present, the 2-SiSP distance is the same $d_2(a_i,q)=3$. If the edge is not present, $d_2(a_i,q) \ge 5$ using a path segment of the form $a_i$-$c$-$a_j$-$c'$. Since only 2 edges cross the cut separating vertices controlled by Alice and Bob, they can simulate a 2-APSiSP algorithm taking $R(n)$ rounds using $O(R(n)\cdot \log n)$ bits. After computing all 2-APSiSP distances, Bob can identify all bits of $S_a$ using the distances $d_2(a_i,q)$ computed at $q$. Then, Bob can readily evaluate set disjointness of $S_a,S_b$ and forward the 1-bit result to Alice in one round. So, the communication lower bound for set disjointness~\cite{Bar-YossefJKS04} gives an $\Omega(n/\log n)$ lower bound for 2-APSiSP. We also claim that APSP distances between vertices controlled by Alice and Bob are the same regardless of bits $S_a,S_b$. Only distances involving $a_i$ could be affected by the input bits, and regardless of the presence of edge $(a_i,c')$ we have $d(a_i,p)=2$ and $d(a_i,q)=3$. 

    We can generalize this lower bound to hold for $(3-\epsilon)$-approximation algorithms (for any constant $\epsilon>0$) by replacing the edges $(a_i,c)$ and $(a_i,c')$ with paths of length $k$, for a large constant $k$. The graph has $\Theta(nk)$ vertices, and we have $d_2(a_i,q)=k+2$. If edge $(a_i,c')$ is not present, then $d_2(a_i,q) \ge 3k+2$, allowing us to determine set disjointness even if only a $(3-\epsilon)$-approximation of $d_2(a_i,q)$ is known (by choosing $k=\Theta(1/\epsilon)$). For directed graphs, we can direct all edges from left to right, so there is no 2-SiSP path from $a_i$ to $q$ when the edge $(a_i,c')$ is not present and $d_2(a_i,q)=\infty$. In this case, the lower bound holds for an arbitrarily large approximation factor.
\end{proof}

\section{Conclusion and Open Problems}
\label{sec:conclusion}

In this paper, we present CONGEST upper and lower bounds for DSO and 2-APSiSP. Our bounds for 2-APSiSP are almost optimal, up to polylog factors, for even undirected unweighted graphs. There still remains a gap between upper and lower bounds for directed weighted DSO, leading to the following open problems:

\begin{itemize}
    \item For a DSO that can answer a batch of $k$ queries in $O(k+D)$ rounds, we present a lower bound of $\tilde{\Omega}(n)$ preprocessing rounds, and an upper bound that takes $\tilde{O}(n^{3/2})$ rounds. Can we bridge the gap between these bounds?
    \item For a DSO that takes $o\left(k\frac{\sqrt{n}}{\log n}+D\right)$ rounds to answer $k$ queries, we show a lower bound of $\tilde{\Omega}(\sqrt{n})$ preprocessing rounds. We show an upper bound that takes $\tilde{O}(n)$ preprocessing rounds and answers $k$ queries in $\tilde{O}(k\sqrt{n}+D)$. It is unlikely the lower bound can be improved for algorithms taking $o(k\sqrt{n}+D)$ query rounds, as this would require an improved CONGEST SSSP lower bound. However, it remains an open question whether the $\tilde{O}(n)$ preprocessing is optimal or can be improved to sub-linear rounds.
    \item Can we achieve a tradeoff between preprocessing rounds and query response rounds, such as $\tilde{O}(n^{3/2-c})$ preprocessing rounds to answer $k$ queries in $\tilde{O}(kn^{c}+D)$ rounds, for $0 < c < 1/2$? We achieve this for $c=0$ and $c=1/2$, but our techniques do not immediately extend to arbitrary $c$.
    \item For seb-DSO, we have shown that with no preprocessing the algorithm using $k$-SSP is optimal. Can we improve on this bound using preprocessing beyond what we have obtained for general DSO?
    \item Our algorithms are for directed weighted graphs (apart from straightforward upper bounds for undirected seb-DSO without preprocessing based on $k$-SSP~\cite{ElkinN19,LenzenPP19,HoangPDGYPR19}). Can we improve on our general DSO results for undirected or unweighted graphs? Can we do better if only approximate DSO distances are required?
    \item Recently, we have obtained the first non-trivial results for DSO on the PRAM including a work-efficient algorithm with $\tilde{O}(n^{1/2+o(1)})$ parallel time and $\tilde{O}(mn)$ work, an algorithm with $\tilde{O}(1)$ parallel time with $\tilde{O}(n^3)$ work, and an algorithm with a work-time tradeoff that can achieve sub-$n^3$ work and sub-$\sqrt{n}$ time~\cite{ManoharanR25pram}. Further improvements for the PRAM model and the broader study of DSO in other parallel settings 
    are topics for further investigation.
\end{itemize}

\bibliographystyle{plainurl}
\bibliography{references} 

@inproceedings{ManoharanR25pram,
  author       = {Vignesh Manoharan and
                  Vijaya Ramachandran},
  title        = {Brief Announcement: Algorithms for Distance Sensitivity Oracles on
                  the {PRAM}},
  booktitle    = {Proceedings of the 37th {ACM} Symposium on Parallelism in Algorithms
                  and Architectures, {SPAA} 2025, Portland, OR, USA, 28 July 2025 -
                  1 August 2025},
  pages        = {628--632},
  publisher    = {{ACM}},
  year         = {2025}
}

@inproceedings{ManoharanR24rp,
  author       = {Vignesh Manoharan and
                  Vijaya Ramachandran},
  title        = {Computing Replacement Paths in the {CONGEST} Model},
  booktitle    = {Structural Information and Communication Complexity - 31st International
                  Colloquium, {SIROCCO} 2024, Vietri sul Mare, Italy, May 27-29, 2024,
                  Proceedings},
  series       = {Lecture Notes in Computer Science},
  volume       = {14662},
  pages        = {420--437},
  publisher    = {Springer},
  year         = {2024}
}

@inproceedings{ManoharanR24mwc,
  author       = {Vignesh Manoharan and
                  Vijaya Ramachandran},
  title        = {Computing Minimum Weight Cycle in the {CONGEST} Model},
  booktitle    = {Proceedings of the 43rd {ACM} Symposium on Principles of Distributed
                  Computing, {PODC} 2024, Nantes, France, June 17-21, 2024},
  pages        = {182--193},
  publisher    = {{ACM}},
  year         = {2024}
}

@inproceedings{ManoharanR25dso,
  author       = {Vignesh Manoharan and
                  Vijaya Ramachandran},
  title        = {Distributed Distance Sensitivity Oracles},
  booktitle    = {Structural Information and Communication Complexity - 32nd International
                  Colloquium, {SIROCCO} 2025, Delphi, Greece, June 2-4, 2025, Proceedings},
  series       = {Lecture Notes in Computer Science},
  volume       = {15671},
  pages        = {366--383},
  publisher    = {Springer},
  year         = {2025}
}

@inproceedings{ChangCDMNS25,
  author       = {Yi{-}Jun Chang and
                  Yanyu Chen and
                  Dipan Dey and
                  Gopinath Mishra and
                  Hung Thuan Nguyen and
                  Bryce Sanchez},
  title        = {Optimal Distributed Replacement Paths},
  booktitle    = {Proceedings of the {ACM} Symposium on Principles of Distributed Computing,
                  {PODC} 2025, Hotel Las Brisas Huatulco, Huatulco, Mexico, June 16-20,
                  2025},
  pages        = {287--298},
  publisher    = {{ACM}},
  year         = {2025}
}

@inproceedings{DeyG24,
  author       = {Dipan Dey and
                  Manoj Gupta},
  title        = {Nearly Optimal Fault Tolerant Distance Oracle},
  booktitle    = {Proceedings of the 56th Annual {ACM} Symposium on Theory of Computing,
                  {STOC} 2024, Vancouver, BC, Canada, June 24-28, 2024},
  pages        = {944--955},
  publisher    = {{ACM}},
  year         = {2024}
}

@article{BodwinP23,
  author       = {Greg Bodwin and
                  Merav Parter},
  title        = {Restorable Shortest Path Tiebreaking for Edge-Faulty Graphs},
  journal      = {J. {ACM}},
  volume       = {70},
  number       = {5},
  pages        = {28:1--28:24},
  year         = {2023}
}

@inproceedings{Parter22vertex,
  author       = {Merav Parter},
  title        = {Nearly optimal vertex fault-tolerant spanners in optimal time: sequential,
                  distributed, and parallel},
  booktitle    = {{STOC} '22: 54th Annual {ACM} {SIGACT} Symposium on Theory of Computing,
                  Rome, Italy, June 20 - 24, 2022},
  pages        = {1080--1092},
  publisher    = {{ACM}},
  year         = {2022}
}

@inproceedings{DinitzR20,
  author       = {Michael Dinitz and
                  Caleb Robelle},
  title        = {Efficient and Simple Algorithms for Fault-Tolerant Spanners},
  booktitle    = {{PODC} '20: {ACM} Symposium on Principles of Distributed Computing,
                  Virtual Event, Italy, August 3-7, 2020},
  pages        = {493--500},
  publisher    = {{ACM}},
  year         = {2020}
}

@article{LenzenPP19,
  author       = {Christoph Lenzen and
                  Boaz Patt{-}Shamir and
                  David Peleg},
  title        = {Distributed distance computation and routing with small messages},
  journal      = {Distributed Comput.},
  volume       = {32},
  number       = {2},
  pages        = {133--157},
  year         = {2019}
}

@inproceedings{Nanongkai14,
  author       = {Danupon Nanongkai},
  title        = {Distributed approximation algorithms for weighted shortest paths},
  booktitle    = {Symposium on Theory of Computing, {STOC} 2014, New York, NY, USA,
                  May 31 - June 03, 2014},
  pages        = {565--573},
  publisher    = {{ACM}},
  year         = {2014}
}

@inproceedings{CHLP21oracle,
  author       = {Keren Censor{-}Hillel and
                  Dean Leitersdorf and
                  Volodymyr Polosukhin},
  title        = {On Sparsity Awareness in Distributed Computations},
  booktitle    = {{SPAA} '21: 33rd {ACM} Symposium on Parallelism in Algorithms and
                  Architectures, Virtual Event, USA, 6-8 July, 2021},
  pages        = {151--161},
  publisher    = {{ACM}},
  year         = {2021}
}

@inproceedings{HoangPDGYPR19,
  author       = {Loc Hoang and
                  Matteo Pontecorvi and
                  Roshan Dathathri and
                  Gurbinder Gill and
                  Bozhi You and
                  Keshav Pingali and
                  Vijaya Ramachandran},
  title        = {A round-efficient distributed betweenness centrality algorithm},
  booktitle    = {Proceedings of the 24th {ACM} {SIGPLAN} Symposium on Principles and
                  Practice of Parallel Programming, PPoPP 2019, Washington, DC, USA,
                  February 16-20, 2019},
  pages        = {272--286},
  publisher    = {{ACM}},
  year         = {2019}
}

@book{Peleg2000book,
author = {Peleg, David},
title = {Distributed computing: a locality-sensitive approach},
year = {2000},
isbn = {0898714648},
publisher = {Society for Industrial and Applied Mathematics},
address = {USA}
}

@inproceedings{Ghaffari15scheduling,
  author       = {Mohsen Ghaffari},
  title        = {Near-Optimal Scheduling of Distributed Algorithms},
  booktitle    = {Proceedings of the 2015 {ACM} Symposium on Principles of Distributed
                  Computing, {PODC} 2015, Donostia-San Sebasti{\'{a}}n, Spain,
                  July 21 - 23, 2015},
  pages        = {3--12},
  publisher    = {{ACM}},
  year         = {2015}
}

@article{LeightonMR99,
  author       = {Frank Thomson Leighton and
                  Bruce M. Maggs and
                  Andr{\'{e}}a W. Richa},
  title        = {Fast Algorithms for Finding O(Congestion + Dilation) Packet Routing
                  Schedules},
  journal      = {Comb.},
  volume       = {19},
  number       = {3},
  pages        = {375--401},
  year         = {1999}
}

@article{ElkinN19,
  author       = {Michael Elkin and
                  Ofer Neiman},
  title        = {Hopsets with Constant Hopbound, and Applications to Approximate Shortest
                  Paths},
  journal      = {{SIAM} J. Comput.},
  volume       = {48},
  number       = {4},
  pages        = {1436--1480},
  year         = {2019}
}

@inproceedings{CaoF23parallel,
  author       = {Nairen Cao and
                  Jeremy T. Fineman},
  title        = {Parallel Exact Shortest Paths in Almost Linear Work and Square Root
                  Depth},
  booktitle    = {Proceedings of the 2023 {ACM-SIAM} Symposium on Discrete Algorithms,
                  {SODA} 2023, Florence, Italy, January 22-25, 2023},
  pages        = {4354--4372},
  publisher    = {{SIAM}},
  year         = {2023}
}

@article{BernsteinN19apsp,
  author       = {Aaron Bernstein and
                  Danupon Nanongkai},
  title        = {Distributed Exact Weighted All-Pairs Shortest Paths in Randomized
                  Near-Linear Time},
  journal      = {{SIAM} J. Comput.},
  volume       = {52},
  number       = {on},
  pages        = {STOC19--112--STOC19--127},
  year         = {2019}
}

@article{SarmaHKKNPPW12,
  author       = {Atish Das Sarma and
                  Stephan Holzer and
                  Liah Kor and
                  Amos Korman and
                  Danupon Nanongkai and
                  Gopal Pandurangan and
                  David Peleg and
                  Roger Wattenhofer},
  title        = {Distributed Verification and Hardness of Distributed Approximation},
  journal      = {{SIAM} J. Comput.},
  volume       = {41},
  number       = {5},
  pages        = {1235--1265},
  year         = {2012}
}

@inproceedings{GhaffariP16fault,
  author       = {Mohsen Ghaffari and
                  Merav Parter},
  title        = {Near-Optimal Distributed Algorithms for Fault-Tolerant Tree Structures},
  booktitle    = {Proceedings of the 28th {ACM} Symposium on Parallelism in Algorithms
                  and Architectures, {SPAA} 2016, Asilomar State Beach/Pacific Grove,
                  CA, USA, July 11-13, 2016},
  pages        = {387--396},
  publisher    = {{ACM}},
  year         = {2016}
}

@inproceedings{BernsteinK08dso,
  author       = {Aaron Bernstein and
                  David R. Karger},
  title        = {Improved distance sensitivity oracles via random sampling},
  booktitle    = {Proceedings of the Nineteenth Annual {ACM-SIAM} Symposium on Discrete
                  Algorithms, {SODA} 2008, San Francisco, California, USA, January 20-22,
                  2008},
  pages        = {34--43},
  publisher    = {{SIAM}},
  year         = {2008}
}

@article{Yen1971,
  title={Finding the k shortest loopless paths in a network},
  author={Yen, Jin Y},
  journal={Management Science},
  volume={17},
  number={11},
  pages={712--716},
  year={1971},
  publisher={Informs}
}

@article{Bar-YossefJKS04,
  author       = {Ziv Bar{-}Yossef and
                  T. S. Jayram and
                  Ravi Kumar and
                  D. Sivakumar},
  title        = {An information statistics approach to data stream and communication
                  complexity},
  journal      = {J. Comput. Syst. Sci.},
  volume       = {68},
  number       = {4},
  pages        = {702--732},
  year         = {2004}
}

@inproceedings{AgarwalR16sisp,
  author       = {Udit Agarwal and
                  Vijaya Ramachandran},
  title        = {Finding k Simple Shortest Paths and Cycles},
  booktitle    = {27th International Symposium on Algorithms and Computation, {ISAAC}
                  2016, December 12-14, 2016, Sydney, Australia},
  series       = {LIPIcs},
  volume       = {64},
  pages        = {8:1--8:12},
  publisher    = {Schloss Dagstuhl - Leibniz-Zentrum f{\"{u}}r Informatik},
  year         = {2016}
}

@inproceedings{AgarwalR18finegrained,
  author       = {Udit Agarwal and
                  Vijaya Ramachandran},
  title        = {Fine-grained complexity for sparse graphs},
  booktitle    = {Proceedings of the 50th Annual {ACM} {SIGACT} Symposium on Theory
                  of Computing, {STOC} 2018, Los Angeles, CA, USA, June 25-29, 2018},
  pages        = {239--252},
  publisher    = {{ACM}},
  year         = {2018}
}

@inproceedings{AgarwalR20apsp,
  author       = {Udit Agarwal and
                  Vijaya Ramachandran},
  title        = {Faster Deterministic All Pairs Shortest Paths in Congest Model},
  booktitle    = {{SPAA} '20: 32nd {ACM} Symposium on Parallelism in Algorithms and
                  Architectures, Virtual Event, USA, July 15-17, 2020},
  pages        = {11--21},
  publisher    = {{ACM}},
  year         = {2020}
}

@inproceedings{GhaffariT24sssp,
  author       = {Mohsen Ghaffari and
                  Anton Trygub},
  title        = {A Near-Optimal Low-Energy Deterministic Distributed {SSSP} with Ramifications
                  on Congestion and {APSP}},
  booktitle    = {Proceedings of the 43rd {ACM} Symposium on Principles of Distributed
                  Computing, {PODC} 2024, Nantes, France, June 17-21, 2024},
  pages        = {401--411},
  publisher    = {{ACM}},
  year         = {2024}
}

@article{DemetrescuTCR08dso,
  author       = {Camil Demetrescu and
                  Mikkel Thorup and
                  Rezaul Alam Chowdhury and
                  Vijaya Ramachandran},
  title        = {Oracles for Distances Avoiding a Failed Node or Link},
  journal      = {{SIAM} J. Comput.},
  volume       = {37},
  number       = {5},
  pages        = {1299--1318},
  year         = {2008}
}

@inproceedings{BernsteinK09dso,
  author       = {Aaron Bernstein and
                  David R. Karger},
  title        = {A nearly optimal oracle for avoiding failed vertices and edges},
  booktitle    = {Proceedings of the 41st Annual {ACM} Symposium on Theory of Computing,
                  {STOC} 2009, Bethesda, MD, USA, May 31 - June 2, 2009},
  pages        = {101--110},
  publisher    = {{ACM}},
  year         = {2009}
}

@article{WeimannY13,
  author       = {Oren Weimann and
                  Raphael Yuster},
  title        = {Replacement Paths and Distance Sensitivity Oracles via Fast Matrix
                  Multiplication},
  journal      = {{ACM} Trans. Algorithms},
  volume       = {9},
  number       = {2},
  pages        = {14:1--14:13},
  year         = {2013}
}

@article{Ren22,
  author       = {Hanlin Ren},
  title        = {Improved distance sensitivity oracles with subcubic preprocessing
                  time},
  journal      = {J. Comput. Syst. Sci.},
  volume       = {123},
  pages        = {159--170},
  year         = {2022}
}

@inproceedings{HolzerW12,
  author       = {Stephan Holzer and
                  Roger Wattenhofer},
  title        = {Optimal distributed all pairs shortest paths and applications},
  booktitle    = {{ACM} Symposium on Principles of Distributed Computing, {PODC} '12,
                  Funchal, Madeira, Portugal, July 16-18, 2012},
  pages        = {355--364},
  publisher    = {{ACM}},
  year         = {2012}
}

@article{DemetrescuI04,
  author       = {Camil Demetrescu and
                  Giuseppe F. Italiano},
  title        = {A new approach to dynamic all pairs shortest paths},
  journal      = {J. {ACM}},
  volume       = {51},
  number       = {6},
  pages        = {968--992},
  year         = {2004}
}

\end{document}